\documentclass[thm-restate,a4paper,backref]{lipics-v2021}
\pdfoutput=1
\usepackage{style}
\usepackage{shortcuts}
\usepackage[shortcuts]{extdash}
\usepackage[capitalize]{cleveref}
\crefname{claim}{Claim}{Claims}
\Crefname{claim}{Claim}{Claims}
\crefname{observation}{Observation}{Observations}
\Crefname{observation}{Observation}{Observations}
\renewcommand\epsilon{\varepsilon}

\nolinenumbers
\hideLIPIcs

\title{Negative-Weight Single-Source Shortest Paths in Near-Linear Time: Now Faster!}
\titlerunning{Negative-Weight Single-Source Shortest Paths in Near-Linear Time: Now Faster!}
\author{Karl Bringmann}
    {Saarland University and Max Planck Institute for Informatics,\\Saarland Informatics Campus, Saarbrücken, Germany}{}{}{}
\author{Alejandro Cassis}
    {Saarland University and Max Planck Institute for Informatics,\\Saarland Informatics Campus, Saarbrücken, Germany}{}{}{}
\author{Nick Fischer}
    {Weizmann Institute of Science, Rehovot, Israel}{}{}{Parts of this work were done while the author was at Saarland University and Max Planck Institute for Informatics. This work is part of the project CONJEXITY that has received funding from the European Research Council (ERC) under the European Union's Horizon Europe research and innovation programme (grant agreement No.~101078482).}
\authorrunning{K. Bringmann, A. Cassis, N. Fischer}
\Copyright{Karl Bringmann, Alejandro Cassis, Nick Fischer}
\funding{This work is part of the project TIPEA that has received funding from the European Research Council (ERC) under the European Unions Horizon 2020 research and innovation programme (grant agreement No.~850979).}
\acknowledgements{We thank Danupon Nanongkai for many helpful discussions on the topic.}

\begin{document}
\maketitle

\begin{abstract}
In this work we revisit the fundamental Single-Source Shortest Paths (SSSP) problem with possibly negative edge weights. A recent breakthrough result by Bernstein, Nanongkai and Wulff-Nilsen established a near-linear $O(m \log^8(n) \log(W))$-time algorithm for negative-weight SSSP, where $W$ is an upper bound on the magnitude of the smallest negative-weight edge. In this work we improve the running time to $O(m \log^2(n) \log(nW) \log\log n)$, which is an improvement by nearly six log-factors. Some of these log-factors are easy to shave (e.g.\ replacing the priority queue used in Dijkstra's algorithm), while others are significantly more involved (e.g.\ to find negative cycles we design an algorithm reminiscent of noisy binary search and analyze it with drift analysis).

As side results, we obtain an algorithm to compute the minimum cycle mean in the same running time as well as a new construction for computing Low-Diameter Decompositions in directed graphs.
\end{abstract}

\section{Introduction} \label{sec:introduction}
For centuries, the quest to compute shortest paths has captivated the minds of mathematicians and computer scientists alike. In this work, we revisit the Single-Source Shortest Paths (SSSP) problem with possibly negative edge weights: Given a directed weighted graph $G$ and a designated source vertex $s$, compute the distances from $s$ to all other vertices in~$G$. This is possibly the most fundamental weighted graph problem with wide-ranging applications in computer science, including routing, data networks, artificial intelligence, planning, and operations research.

While it is well known for almost 60 years that SSSP with nonnegative edge weights can be solved in near-linear time by Dijkstra's algorithm~\cite{Dijkstra59,Williams64,FredmanT84}, the case with negative weights has a more intriguing history: The Bellman-Ford algorithm was developed in the~50's~\cite{Shimbel55,Ford56,Bellman58,Moore59} and runs in time $\Order(mn)$, and this time bound remained the state of the art for a long time. Starting in the 80's, the \emph{scaling technique} was developed and lead to time $\Order(m \sqrt n \log W)$~\cite{Gabow83,GabowT89,Goldberg95}; here and throughout,~$W$~is the magnitude of the smallest negative edge weight in the graph.\footnote{Strictly speaking, $W \geq 0$ is the smallest number such that all edge weights satisfy $w(e) \geq -W$. By slight abuse of notation we typically write $\Order(\log W)$ to express $\Order(\max\set{1, \log W})$.} Other papers focused on specialized graph classes, leading e.g.\ to near-linear time algorithms for planar directed graphs~\cite{LiptonRT79,KleinRRS94,FakcharoenpholR01,KleinMW09}, and improved algorithms for dense graphs with small weights~\cite{Sankowski05}.

An alternative approach is to model SSSP as a minimum-cost flow problem.\footnote{To model SSSP as a minimum-cost flow problem, interpret each edge $e$ with weight $w(e)$ as an edge with infinite capacity and cost $w(e)$. Moreover, add an artificial sink vertex $t$ to the graph, and add unit-capacity cost-$0$ edges from all vertices $v$ to $t$. Then any minimum-cost flow routing $n$ units from $s$ to $t$ corresponds exactly to a shortest path tree in the original graph (assuming that it does not contain a negative-weight cycle).} In the last decade, a combination of convex optimization techniques and dynamic algorithms have resulted in a series of advancements in minimum-cost flow computations~\cite{CohenMSV17,AxiotisMV20,BrandLNPSS0W20,BrandLLSS0W21} and thus also for negative-weight SSSP, with running times $\widetilde\Order(m^{10/7})$~\cite{CohenMSV17}, $\widetilde\Order(m^{4/3})$~\cite{AxiotisMV20} and~\makebox{$\widetilde\Order(m + n^{3/2})$}~\cite{BrandLNPSS0W20}. This line of research recently culminated in an almost-linear $m^{1+\order(1)}$-time algorithm by Chen, Kyng, Liu, Peng, Probst Gutenberg and Sachdeva~\cite{ChenKLPGS22}.

Finally, at the same time as the breakthrough in computing minimum-cost flows, Bernstein, Nanongkai and Wulff-Nilsen~\cite{BernsteinNW22} found an astonishing \emph{near-linear} $\Order(m \log^8(n) \log(W))$-time algorithm for negative-weight SSSP. We will refer to their algorithm as the \emph{BNW algorithm}. The BNW algorithm is combinatorial and arguably simple, and thus a satisfying answer to the coarse-grained complexity of the negative-weight SSSP problem. However, the story does not end here. In this work, we press further and investigate the following question which was left open by Bernstein et al.~\cite{BernsteinNW22}:

\bigskip\noindent
\parbox{\textwidth}{\centering\emph{Can we further improve the number of log-factors \\ in the running time of negative-weight SSSP?}}

\bigskip
For comparison, the \emph{nonnegative}-weights SSSP problem underwent a long series of lower-order improvements in the last century~\cite{Dijkstra59,FredmanT84,FredmanW93,FredmanW94,Thorup00,Raman96,Raman97,Boas77,BoasKZ77,MehlhornN90,AhujaMOT90,CherkasskyGS99,Thorup03}, including improvements by log-factors or even just loglog-factors.\footnote{In these papers, the Dijkstra running time $\Order(m + n \log n)$ was improved to the current state of the art~$\Order(m + n \log\log \min\{n, C\})$~\cite{Thorup03}, where~$C$ is the largest weight in the graph.} In the same spirit, we initiate the fine-grained study of lower-order factors for negative-weight shortest paths.

\subsection{Our Results}
In our main result we make significant progress on our driving question, and improve the BNW algorithm by nearly six log-factors:

\begin{restatable}[Negative-Weight SSSP]{theorem}{thmssspornegativecycle} \label{thm:sssp-or-negative-cycle}
There is a Las Vegas algorithm which, given a directed graph $G$ and a source node $s$, either computes a shortest path tree from $s$ or finds a negative cycle in $G$, running in time $\Order((m + n \log \log n) \log^2 n \log(nW))$ with high probability (and in expectation).
\end{restatable}

\smallskip
We obtain this result by optimizing the BNW algorithm, pushing it to its limits. Aaron Bernstein remarked in a presentation of their result that ``something like $\log^5$ is inherent to [their] current framework''.\footnote{\url{https://youtu.be/Bpw3yqWT_d0?t=3721}} It is thus surprising that we obtain such a dramatic improvement to nearly three log-factors within the same broader framework.
Despite this speed-up, our algorithm is still modular and simple in its core. In \cref{sec:introduction:sec:technical-contribution} we discuss the technical similarities and differences between our algorithm and the BNW algorithm in detail.

Recall that computing shortest paths is only reasonable in graphs without negative cycles (as otherwise two nodes are possibly connected by a path of arbitrarily small negative weight). In light of this, we solve the negative-weight SSSP problem in its strongest possible form in \cref{thm:sssp-or-negative-cycle}: The algorithm either returns a shortest path tree, or returns a negative cycle as a certificate that no shortest path tree exists. In fact, the subproblem of detecting negative cycles has received considerable attention on its own in the literature (see e.g.\ the survey~\cite{CherkasskyG99}).

In the presence of negative cycles, a natural alternative to finding one such cycle is to compute all distances in the graph anyway (where some of the distances are $-\infty$ or~$\infty$). This task can be solved in the same running time:

\begin{restatable}[Negative-Weight Single-Source Distances]{theorem}{thmalldistances} \label{thm:all-distances}
There is a Las Vegas algorithm, which, given a directed graph $G$ and a source $s \in V(G)$, computes the distances from $s$ to all other vertices in the graph (these distances are possibly $-\infty$ or $\infty$), running in time $\Order((m + n \log \log n) \log^2 n \log(nW))$ with high probability (and in expectation).
\end{restatable}

Owing to the countless practical applications of shortest paths problems, it is an important question to ask whether there is a negative-weights shortest paths algorithm that has a competitive implementation. The typical practical implementation in competitive programming uses optimized variants of Bellman-Ford's algorithm, such as the ``Shortest Path Faster Algorithm''~\cite{Moore59,Wikipedia23} (see also~\cite{CherkasskyGGTW08} for an experimental evaluation of other more sophisticated variants of Bellman-Ford). However, it is easy to find instances for which these algorithms require time $\Omega(mn)$. It would be exciting if, after decades of competitive programming, there finally was an efficient implementation to deal with these instances. With its nine log-factors, the BNW algorithm does not qualify as a practical candidate. We believe that our work paves the way for a comparably fast implementation.

\medskip
In addition to our main result, we make progress on two closely related problems: Computing the minimum cycle mean, and low-diameter decompositions in directed graphs. We describe these results in the following sections.

\subsubsection{Minimum Cycle Mean}
In a directed weighted graph, the \emph{mean} of a cycle $C$ is defined as the ratio $\bar w(C) = w(C) / |C|$ where $w(C)$ is the total weight of $C$. The \emph{Minimum Cycle Mean} problem is to compute, in a given directed weighted graph, the minimum mean across all cycles, $\min_C \bar w(C)$. This is a central problem in the context of network flows~\cite{AhujaMO93}, with applications to verification and reactive systems analysis~\cite{ChatterjeeHKLR14}.

There is a large body of literature on computing the Minimum Cycle Mean. In 1987, Karp~\cite{Karp78} established an $\Order(mn)$-time algorithm, which is the fastest strongly polynomial time algorithm to this date. In terms of weakly polynomial algorithms, Lawler observed that the problem is equivalent to detecting negative cycles, up to a factor $\Order(\log(nW))$~\cite{Lawler66,Lawler76}. Indeed, note that one direction is trivial: The graph has a negative cycle if and only if the minimum cycle mean is negative. For the other direction, he provided a reduction to detecting negative cycles on $\Order(\log(nW))$ graphs with modified rational edge weights. Thus, following Lawler's observation, any negative-weight SSSP algorithm can be turned into a Minimum Cycle Mean algorithm in a black-box way with running time overhead $\Order(\log(nW))$.

There are also results specific to Minimum Cycle Mean computations: Orlin and Ahuja~\cite{OrlinA92} designed an algorithm in time $\Order(m \sqrt n \log(nW))$ (improving over the baseline $\Order(m \sqrt n \log^2(nW))$ which follows from the SSSP algorithms by~\cite{Gabow83,GabowT89,Goldberg95}). For the special case of dense graphs with $0$-$1$-weights, an $\Order(n^2)$-time algorithm is known~\cite{ButkovicC92}. Finally, in terms of approximation algorithms it is known how to compute a $(1+\varepsilon)$\=/approximation in time $\widetilde\Order(n^\omega \log(W) / \varepsilon)$~\cite{ChatterjeeHKLR14}.

As for negative-weight SSSP, all these algorithms are dominated by the recent BNW algorithm: By Lawler's observation, their algorithm computes the minimum cycle mean in time $\Order(m \log^8(n) \log^2(nW))$. In fact, it is implicit in their work that the running time can be reduced to $\Order(m \log^8(n) \log(nW))$. Our contribution is again that we reduce the number of log-factors from nine to nearly three:

\begin{restatable}[Minimum Cycle Mean]{theorem}{thmmincyclemean}\label{thm:mincyclemean}
There is a Las Vegas algorithm, which given a directed graph $G$ finds a cycle $C$ with minimum mean weight $\bar w(C) = \min_{C'} \bar w(C')$, running in time~$\Order((m + n \log \log n) \log^2 n \log(nW))$ with high probability (and in expectation).
\end{restatable}

\subsubsection{Directed Low-Diameter Decompositions}
A crucial ingredient to the BNW algorithm is a Low-Diameter Decomposition (LDD) in directed graphs. Our SSSP algorithm differs in that regard from the BNW algorithm, and does not explicitly use LDDs. Nevertheless, as a side result of this work we improve the best known LDD in directed graphs.

LDDs have been first studied by Awerbuch almost 40 years ago~\cite{Awerbuch85} and have ever since found several applications, mostly for undirected graphs and mostly in distributed, parallel and dynamic settings~\cite{AwerbuchGLP89,AwerbuchP92,AwerbuchBCP92,LinialS93,Bartal96,BlellochGKMPT14,MillerPX13,PachockiRSTW18,ForsterG19,ChechikZ20,BernsteinGW20,ForsterGV21,BernsteinNW22}. The precise definitions in these works mostly differ, but the common idea is to select a small subset of edges $S$ such that after removing all edges in $S$ from the graph, the remaining graph has (strongly) connected components with bounded diameter.

For directed graphs, we distinguish two types of LDDs: \emph{Weak} LDDs ensure that for every strongly connected component $C$ in the graph $G \setminus S$, the diameter of $C$ \emph{in the original graph} is bounded. A \emph{strong} LDD exhibits the stronger property that the diameter of $C$ in the graph $G \setminus S$ is bounded.

\begin{definition}[Directed Low-Diameter Decomposition] \label{def:ldd}
A \emph{weak Low-Diameter Decomposition} with overhead $\rho$ is a Las Vegas algorithm that, given a directed graph $G$ with nonnegative edge weights $w$ and a parameter $D > 0$, computes an edge set $S \subseteq E(G)$ with the following properties:
\smallskip
\begin{itemize}
    \item \emph{Sparse Hitting:} For any edge $e \in E$, $\Pr(e \in S) \leq \Order(\frac{w(e)}D \cdot \rho + \frac{1}{\poly(n)})$.
    \item \emph{Weak Diameter:} Every SCC $C$ in $G \setminus S$ has weak diameter at most $D$ \\(that is, for any two vertices $u, v \in C$, we have $\dist_G(u, v) \leq D$).
\end{itemize}
\smallskip
We say that the Low-Diameter Decomposition is \emph{strong} if it additionally satisfies the following stronger property:
\smallskip
\begin{itemize}
    \item \emph{Strong Diameter:} Every SCC $C$ in $G \setminus S$ has diameter at most~$D$\\(that is, for any two vertices $u, v \in C$, we have $\dist_{G \setminus S}(u, v) \leq D$).
\end{itemize}
\end{definition}

For directed graphs, the state-of-the-art \emph{weak} LDD was developed by Bernstein, Nanongkai and Wulff-Nilsen~\cite{BernsteinNW22} as a tool for their shortest paths algorithm. Their result is a weak LDD with polylogarithmic overhead~$\Order(\log^2 n)$ running in near-linear time $\Order(m \log^2 n + n \log^2 n \log\log n)$. In terms of \emph{strong} LDDs, no comparable result is known. While it is not hard to adapt their algorithm to compute a strong LDD, this augmentation suffers from a slower running time $\Omega(nm)$. Our contribution is designing the first strong LDD computable in near-linear time, with only slightly worse overhead $\Order(\log^3 n)$:

\begin{restatable}[Strong Low-Diameter Decomposition]{theorem}{thmstrongldd} \label{thm:strong-ldd}
There is a strong Low-Diameter Decomposition with overhead $\Order(\log^3 n)$, computable in time $\Order((m + n \log \log n) \log^2 n)$ with high probability (and in expectation).
\end{restatable}

\subsection{Technical Overview} \label{sec:introduction:sec:technical-contribution}
Our algorithm is inspired by BNW algorithm and follows its general framework, but differs in many aspects. In this section we give a detailed comparison.

\subsubsection{The Framework}
The presentation of our algorithm is modular: We will first focus on the SSSP problem on a restricted class of graphs (to which we will simply refer as \emph{restricted} graphs, see the next \cref{def:restricted}). In the second step we demonstrate how to obtain our results for SSSP on general graphs, for finding negative cycles, and for computing the minimum cycle mean by reducing to the restricted problem in a black-box manner.

\begin{restatable}[Restricted Graphs]{definition}{defrestricted} \label{def:restricted}
An edge-weighted directed graph $G = (V, E, w)$ with a designated source vertex $s \in V$ is \emph{restricted} if it satisfies:
\begin{itemize}
\item The edge weights are integral and at least $-1$.
\item The minimum cycle mean is at least $1$.
\item The source $s$ is connected to every other vertex by an edge of weight $0$.
\end{itemize}
\end{restatable}

\medskip
In particular, note that restricted graphs do not contain negative cycles, and therefore it is always possible to compute a shortest path tree. The \emph{Restricted SSSP} problem is to compute a shortest path tree in a given restricted graph $G$. We write $T_{\textsc{RSSSP}}(m, n)$ for the optimal running time of a Restricted SSSP algorithm with error probability $\frac12$, say.

\subsubsection{Improvement 1: Faster Restricted SSSP via Better Decompositions}
Bernstein et al.~\cite{BernsteinNW22} proved that $T_{\TRestrictedSSSP}(m, n) = \Order(m \log^5 n)$. Our first contribution is that we shave nearly three log-factors and improve this bound to~$T_{\TRestrictedSSSP}(m, n) = \Order((m + n \log \log n) \log^2 n)$ (see~\cref{thm:restricted-sssp}).

At a high level, the idea of the BNW algorithm is to decompose the graph by finding a subset of edges $S$ suitable for the following two subtasks: (1) We can recursively compute shortest paths in the graph $G \setminus S$ obtained by removing the edges in $S$, and thereby make enough progress to incur in total only a small polylogarithmic overhead in the running time. And~(2), given the outcome of the recursive call, we can efficiently ``add back'' the edges from $S$ to obtain a correct shortest path tree for $G$. For the latter task, the crucial property is that $S$ intersects every shortest path in $G$ at most~$\Order(\log n)$ times (in expectation), as then a simple generalization of Dijkstra's and Bellman-Ford's algorithm can adjust the shortest path tree in near-linear time (see~\cref{lem:lazy-dijkstra}).

For our result, we keep the implementation of step (2) mostly intact, except that we use a faster implementation of Dijkstra's algorithm due to Thorup~\cite{Thorup03} (see~\cref{lem:lazy-dijkstra}). The most significant difference takes place in step (1), where we change how the algorithm selects~$S$. Specifically, Bernstein et al.\ used a directed Low-Diameter Decomposition to implement the decomposition. We are following the same thought, but derive a more efficient and direct decomposition scheme. To this end, we define the following key parameter:

\begin{definition} \label{def:kappa}
Let $G$ be a restricted graph with designated source $s$. We define $\kappa(G)$ as the maximum number of negative edges (that is, edges of weight exactly~$-1$) in any simple path~$P$ which starts at $s$ and has nonpositive weight $w(P) \leq 0$.
\end{definition}

\noindent
Our new decomposition can be stated as follows.

\begin{restatable}[Decomposition]{lemma}{lemdecomposition} \label{lem:decomposition}
Let $G$ be a restricted graph with source vertex $s \in V(G)$ and $\kappa \geq \kappa(G)$. There is a randomized algorithm $\textsc{Decompose}(G, \kappa)$ running in expected time $\Order((m + n \log \log n) \log n)$ that computes an edge set $S \subseteq E(G)$ such that:
\medskip
\begin{enumerate}
    \item \emph{Progress:} With high probability, for any strongly connected component $C$ in $G \setminus S$, we have (i) $|C| \leq \frac34 |V(G)|$ or (ii)~$\kappa(G[C \cup \set s]) \leq \frac\kappa2$.
    \item \emph{Sparse Hitting:} For any shortest $s$-$v$-path $P$ in $G$, we have $\Ex(|P \cap S|) \leq \Order(\log n)$.
\end{enumerate}
\end{restatable}

\medskip
The sparse hitting property is exactly what we need for (2). With the progress condition, we ensure that $|V(G)| \cdot \kappa(G)$ reduces by a constant factor when recurring on the strongly connected components of $G \setminus S$. The recursion tree therefore reaches depth at most~\makebox{$\Order(\log(n \cdot \kappa(G))) = \Order(\log n)$}. In summary, with this new idea we can compute shortest paths in restricted graphs in time $\Order((m + n \log \log n) \log^2 n)$.

\subsubsection{Improvement 2: Faster Scaling}
It remains to lift our Restricted SSSP algorithm to the general SSSP problem at the expense of at most one more log-factor $\log(nW)$. In comparison, the BNW algorithm spends four log-factors $\Order(\log^3 n \log W)$ here. As a warm-up, we assume that the given graph is promised not to contain a negative cycle.

\subparagraph{Warm-Up: From Restricted Graphs to Graphs without Negative Cycles}
This task is a prime example amenable to the \emph{scaling technique} from the 80's~\cite{Gabow83,GabowT89,Goldberg95}: By rounding the weights in the given graph $G$ from $w(e)$ to $\lceil \frac{3 w(e)}{W+1} \rceil + 1$ we ensure that (i) all weights are at least $-1$ and (ii) the minimum cycle mean is at least $1$, and thus we turn $G$ into a restricted graph $H$ (see \cref{lem:scale}). We compute the shortest paths in $H$ and use the computed distances (by means of a \emph{potential function}) to augment the weights in the original graph $G$. If $G$ has smallest weight $-W$, in this way we can obtain an \emph{equivalent} graph $G'$ with smallest weight $- \frac34 W$, where equivalence is defined as follows:

\begin{restatable}[Equivalent Graphs]{definition}{defequivalentgraphs} \label{def:equivalent-graphs}
We say that two graphs $G, G'$ over the same set of vertices and edges are \emph{equivalent} if~(1)~any shortest path in $G$ is also a shortest path in $G'$ and vice versa, and~(2) for any cycle $C$, $w_G(C) = w_{G'}(C)$.
\end{restatable}

\medskip
Hence, by (1) we continue to compute shortest paths in $G'$. At first glance it seems that repeating this scaling step incurs only a factor $\log W$ to the running time, but for subtle reasons the overhead is actually $\log(nW)$. Another issue is that the Restricted SSSP algorithm errs with constant probability. The easy fix loses another $\log n$ factor due to boosting (this is how Bernstein et al.\ obtain their algorithm \textsc{SPMonteCarlo}, see~\cite[Theorem~7.1]{BernsteinNW22}). Fortunately, we can ``merge'' the scaling and boosting steps to reduce the overhead to $\log(nW)$ in total, see~\cref{thm:sssp}.

\subparagraph{From Restricted Graphs to Arbitrary Graphs}
What if $G$ contains a negative cycle? In this case, our goal is to find and return one such negative cycle. Besides the obvious advantage that it makes the output more informative, this also allows us to strengthen the algorithm from Monte Carlo to Las Vegas, since both a shortest path tree and a negative cycle serve as \emph{certificates} that can be efficiently tested. Using the scaling technique as before, we can easily \emph{detect} whether a given graph contains a negative cycle in time $\Order(T_{\TRestrictedSSSP}(m, n) \cdot \log (nW))$ (even with high probability, see~\cref{cor:detect-negcycle}), but we cannot \emph{find} such a cycle.

We give an efficient reduction from finding negative cycles to Restricted SSSP with overhead $\Order(\log(nW))$. This reduction is the technically most involved part of our paper. In the following paragraphs we attempt to give a glimpse into the main ideas.

\subparagraph{A Noisy-Binary-Search-Style Problem}
For the rest of this overview, we phrase our core challenge as abstract as possible, and omit further context.
We will use the following notation: given a directed graph $G$ and an integer $M$, we write $G^{+M}$ to denote
the graph obtained by adding $M$ to every edge weight of $G$.
Consider the following task:

\begin{definition}[Threshold]
Given a weighted graph $G$, compute the smallest integer $M^* \geq 0$ such that the graph $G^{+M^*}$, which is obtained from $G$ by adding $M^*$ to all edge weights, does not contain a negative cycle.
\end{definition}

Our goal is to solve the Threshold problem in time $\Order(T_{\TRestrictedSSSP}(m, n) \log (nW))$ (from this it follows that we can find negative cycles in the same time, see~\cref{lem:reduction-to-threshold}). As a tool, we are allowed to use the following lemma as a black-box (which can be proven similarly to the warm-up case):

\begin{lemma}[Informal \cref{lem:scale-test}] \label{lem:informal-tester}
There is an $\Order(T_{\TRestrictedSSSP}(m, n))$-time algorithm that, given a graph $G$ with minimum weight~$-W$, either returns an equivalent graph $G'$ with minimum weight $-\frac34 W$, or returns~$\textsc{NegativeCycle}$. If $G$ does not contain a negative cycle, then the algorithm returns $\textsc{NegativeCycle}$ with error probability at most $0.01$.
\end{lemma}

Morally, \cref{lem:informal-tester} provides a test whether a given graph $G$ contains a negative cycle. A natural idea is therefore to find $M^*$ by binary search, using \cref{lem:informal-tester} as the tester. However, note that this tester is \emph{one-sided}: If $G$ contains a negative cycle, then the tester is not obliged to detect one. Fortunately, we can turn the tester into a win-win algorithm to compute $M^*$.

We first describe our Threshold algorithm in an idealized setting where we assume that the tester from \cref{lem:informal-tester} has error probability~$0$. We let $d = \frac15 W$, and run the tester on the graph $G^{+d}$. There are two cases:
\begin{itemize}
    \smallskip
    \item\emph{The tester returns $\textsc{NegativeCycle}$:} In the idealized setting we can assume that $G^{+d}$ indeed contains a negative cycle. We therefore compute the threshold of $G^{+d}$ recursively, and return that value plus $d$. Note that the minimum weight of $G^{+d}$ is at least $-W + d = -\frac45 W$.
    \item\emph{The tester returns an equivalent graph $G'$:} In this case, we recursively compute and return the threshold value of $(G')^{-d}$. Note that the graphs $G$ and $(G')^{-d}$ share the same threshold value, as by \cref{def:equivalent-graphs} we have $w_{(G')^{-d}}(C) = w_{G'}(C) - d = w_{G^{+d}}(C) - d = w_G(C)$ for any cycle $C$. Moreover, since $G^{+d}$ has smallest weight $-\frac45 W$, the equivalent graph $G'$ has smallest weight at least $-\frac 34 \cdot \frac 45 W = -\frac35 W$ by \cref{lem:informal-tester}. Therefore, $(G')^{-d}$ has smallest weight at least $-\frac35 W - d = -\frac45 W$.
    \smallskip
\end{itemize}
In both cases, we recursively compute the threshold of a graph with smallest weight at least~$-\frac45 W$. Therefore, the recursion reaches depth~$\Order(\log W)$ until we have reduced the graph to constant minimum weight and the problem becomes easy.

\smallskip
The above algorithm works in the idealized setting, but what about the unidealized setting, where the tester can err with constant probability? We could of course first boost the tester to succeed with high probability. In combination with the above algorithm, this would solve the Threshold problem in time $\Order(T_{\TRestrictedSSSP}(m, n) \log (nW) \log n)$.

However, the true complexity of this task lies in avoiding the naive boosting. By precisely understanding the unidealized setting with constant error probability, we improve the running time for Threshold to $\Order(T_{\TRestrictedSSSP}(m, n) \log (nW))$.
To this end, it seems that one could apply the technique of \emph{noisy binary search} (see e.g.~\cite{Pelc89,FeigeRPU94,Pelc02}). Unfortunately, the known results do not seem applicable to our situation, as \cref{lem:informal-tester} only provides a one-sided tester.
Our solution to this final challenge is an innovative combination of the algorithm sketched above with ideas from noisy binary search. The analysis makes use of \emph{drift analysis} (see e.g.~\cite{Lengler17}), which involves defining a suitable \emph{drift function} (a quantity which in expectation decreases by a constant factor in each step and is zero if and only if we found the optimal value $M^*$) and an application of a \emph{drift theorem} (see~\cref{thm:drift-tailbound}) to prove that the drift function rapidly approaches zero.

\subsection{Summary of Log Shaves}
Finally, to ease the comparison with the BNW algorithm, we compactly list where exactly we shave the nearly six log-factors. We start with the improvements in the Restricted SSSP algorithm:
\begin{itemize}
    \smallskip
    \item We use Thorup's priority queue~\cite{Thorup03} to speed up Dijkstra's algorithm, see~\cref{lem:lazy-dijkstra}. This reduces the cost of one log-factor to a loglog-factor.
    \smallskip
    \item The Sparse Hitting property of our decomposition scheme (\cref{lem:decomposition}) incurs only an $\Order(\log n)$ overhead in comparison to the $\Order(\log^2 n)$ overhead due to the Low-Diameter Decomposition in the BNW algorithm.
    \smallskip
    \item The Progress property of our decomposition scheme (\cref{lem:decomposition}) ensures that the recursion depth of our Restricted SSSP algorithm is just $\Order(\log n)$. The analogous recursion depth in the BNW algorithm is $\Order(\log^2 n)$ (depth $\log n$ for reducing the number of nodes $n$ times depth $\log n$ for reducing $\max_v \eta_G(v)$).
    \smallskip
\end{itemize}
Next, we summarize the log-factors shaved in the scaling step:
\begin{itemize}
    \smallskip
    \item The BNW algorithm amplifies the success probability of the Restricted SSSP algorithm by repeating it $\Order(\log n)$ times. We combine the boosting with the scaling steps which saves this log-factor.
    \smallskip
    \item We improve the overall reduction from finding a negative cycle to Restricted SSSP.
    In particular, we give an implementation of \textsc{Threshold} which is faster by two log-factors (see~\cref{lem:reduction-to-threshold,lem:threshold}). This is where we use an involved analysis via a drift function.
\end{itemize}

\subsection{Open Problems}
Our work leaves open several interesting questions. Can our algorithm for negative-weight Single-Source Shortest Paths be improved further? Specifically:
\smallskip
\begin{enumerate}
\item \emph{Can the number of $\log n$ factors be improved further?}\\In our algorithm, we suffer three log-factors because of (i) the scaling technique ($\log(nW)$) to reduce to restricted graphs, and on restricted graphs (ii) the inherent $\log n$ overhead of the graph decomposition and (iii) the recursion depth $\log n$ to progressively reduce $\kappa(G)$, all of which seem unavoidable. We therefore believe that it is hard to improve upon our algorithm without substantially changing the framework.
\smallskip
\item\emph{Can the loss due to the scaling technique be reduced from $\log(nW)$ to $\log W$?}\\The classical scaling technique, as a reduction to graphs with weights at least $-1$, requires only $\log W$ iterations~\cite{Goldberg95}. But in our setting, due to the stronger conditions for \emph{restricted} graphs (and due to the boosting), we need $\log(nW)$ iterations. Can we do better?
\smallskip
\item\emph{Can the $\log W$ factor be removed from the running time altogether?}\\That is, is there a \emph{strongly polynomial} algorithm in near-linear time? In terms of non-scaling algorithms, the Bellman-Ford algorithm remains state of the art with running time~$\Order(nm)$. This question has been asked repeatedly and appears to be very hard.
\smallskip
\item \emph{Can the algorithm be derandomized?}\\The fastest deterministic algorithm for negative-weight SSSP remains the $\Order(m \sqrt n \log(W))$ algorithm by~\cite{Gabow83,GabowT89,Goldberg95} and it is open to find a near-linear-time algorithm.
\end{enumerate}

\subsection{Outline}
This paper is structured as follows. In \cref{sec:preliminaries} we give some formal preliminaries.
In \cref{sec:restricted-sssp} we present our algorithm for negative-weight SSSP in restricted graphs.
In \cref{sec:sssp-without-neg-cycles} we extend the algorithm from the previous section to work on general graphs without negative cycles.
In \cref{sec:neg-cycle} we remove this assumption and strengthen the algorithm to find negative cycles without worsening the running time.
In \cref{sec:mean-cycle} we give our result for computing the minimum cycle mean.
Finally, in \cref{sec:ldd} we give our improved results for Low-Diameter Decompositions in directed graphs.

\section{Preliminaries} \label{sec:preliminaries}

We write $[n] = \set{1, \dots, n}$ and $\Olog(T) = T \cdot (\log T)^{\Order(1)}$. An event occurs \emph{with high probability} if it occurs with probability $1 - 1/n^c$ for an arbitrarily large constant $c$ (here, $n$ is the number of vertices in the input graph). Unless further specified, our algorithms are Monte Carlo algorithms that succeed with high probability.

\subparagraph{Directed Graphs}
Throughout we consider directed edge-weighted graphs $G = (V, E, w)$. Here $V = V(G)$ is the set of vertices and~\makebox{$E = E(G) \subseteq V(G)^2$} is the set of edges. All edge weights are integers, denoted by $w(e) = w(u, v)$ for~\makebox{$e = (u, v) \in E(G)$}. We typically set~\makebox{$n = |V(G)|$} and~\makebox{$m = |E(G)|$}. We write $G[C]$ to denote the \emph{induced subgraph} with vertices $C \subseteq V(G)$ and write $G \setminus S$ to denote the graph $G$ after deleting all edges in $S \subseteq E$. We write $\deg(v)$ for the (out-)degree of $v$, that is, the number of edges starting from $v$.

A \emph{strongly connected component (SCC)} is a maximal set of vertices $C \subseteq V(G)$ in which all pairs are reachable from each other. It is known that every directed graph can be decomposed into a collection of SCCs, and the graph obtained by compressing the SCCs into single nodes is acyclic. It is also known that the SCCs can be computed in linear time:

\begin{lemma}[Strongly Connected Components, \cite{Tarjan72}] \label{lem:scc}
In any directed graph $G$, the strongly connected components can be identified in time $\Order(n + m)$.
\end{lemma}

For a set of edges $S$ (such as a path or a cycle), we write $w(S) = \sum_{e \in S} w(e)$. A negative cycle is a cycle $C$ with $w(C) < 0$.
For vertices $u, v$, we write $\dist_G(u, v)$ for the length of the shortest $u$-$v$-path. If there is a negative-weight cycle in some $u$-$v$-path, we set $\dist_G(u,v) = -\infty$, and if there is no path from $u$ to $v$ we set $\dist_G(u, v) = \infty$.

\begin{definition}[Balls]\label{def:balls}
    For a vertex $v$, and a nonnegative integer $r$, we denote the out-ball centered at $v$ with radius $r$ by~\makebox{$\Bout_G(v, r) = \set{u \in V(G) : \dist_G(v, u) \leq r}$}.
    Similarly, we denote the in-ball centered at $v$ with radius $r$ by \makebox{$\Bin_G(v, r) = \set{u \in V(G) : \dist_G(u, v) \leq r}$}.
    Further, we write $\partial\Bout_G(v, r) = \set{(u,w) \in E : u \in \Bout(v, r) \land w \notin \Bout(v, r)}$ to denote the boundary edges of an out-ball,
    and $\partial\Bin_G(v, r) = \set{(u,w) \in E : u \notin \Bin(v, r) \land w \in \Bin(v, r)}$ for an in-ball.
\end{definition}

\noindent In all these notations, we occasionally drop the subscript $G$ if it is clear from context.
\smallskip

The \emph{Single-Source Shortest Paths (SSSP)} problem is to compute the distances $\dist_G(s, v)$ for a designated source vertex $s \in V(G)$ to all other vertices $v \in V(G)$. When $G$ does not contain negative cycles, this is equivalent to compute a \emph{shortest path tree} from $s$ (that is, a tree in which every $s$-to-$v$ path is a shortest path in $G$). For graphs with \emph{nonnegative} edge weights, Dijkstra's classical algorithm solves the SSSP problem in near-linear time. We use the following result by Thorup, which replaces the $\log n$ overhead by $\log\log n$ (in the RAM model, see the paragraph on the machine model below).

\begin{restatable}[Dijkstra's Algorithm, \cite{Dijkstra59,Thorup03}]{lemma}{lemdijkstra} \label{lem:dijkstra}
In any directed graph $G$ with nonnegative edge weights, the SSSP problem can be solved in time $\Order(m + n \log\log n)$.
\end{restatable}

\begin{lemma}[Bellman-Ford's Algorithm, \cite{Shimbel55,Ford56,Bellman58,Moore59}] \label{lem:bellman-ford}
In any directed graph $G$, the SSSP problem can be solved in time $\Order(m n)$.
\end{lemma}

\subparagraph{Potentials}
Let $G$ be a directed graph. We refer to functions $\phi : V(G) \to \Int$ as \emph{potential functions}. We write $G_\phi$ for the graph obtained from $G$ by changing the edge weights to~$w_\phi(u, v) = w(u, v) + \phi(u) - \phi(v)$.

\defequivalentgraphs*

\begin{lemma}[Johnson's Trick, \cite{Johnson77}] \label{lem:johnson}
Let $G$ be a directed graph, and let $\phi$ be an arbitrary potential function. Then $w_\phi(P) = w(P) + \phi(u) - \phi(v)$ for any $u$-$v$-path $P$, and~$w_\phi(C) = w(C)$ for any cycle $C$. It follows that $G$ and $G_\phi$ are equivalent.
\end{lemma}

\begin{lemma}[\cite{Johnson77}]\label{cor:johnson-non}
Let $G$ be a directed graph without negative cycles and let $s \in V$ be a source vertex that can reach every other node. Then, for the potential $\phi$ defined as $\phi(v) = \dist_G(s,v)$, it holds that $w_{\phi}(e) \geq 0$ for all edges $e \in E$.
\end{lemma}

\subparagraph{Machine Model}
We work in the standard word RAM model with word size $\Theta(\log n + \log M)$, where~$n$ is the number of vertices and $M$ is an upper bound on the largest edge weight in absolute value. That is, we assume that we can store vertex identifiers and edge weights in a single machine word, and perform basic operations in unit time.
\section{SSSP on Restricted Graphs} \label{sec:restricted-sssp}

In this section we give an efficient algorithm for SSSP on restricted graphs (recall~\cref{def:restricted}). Specifically, we prove the following theorem:

\begin{theorem}[Restricted SSSP] \label{thm:restricted-sssp}
In a restricted graph $G$ with source vertex $s \in V(G)$, we can compute a shortest path tree from $s$ in time $\Order((m + n \log \log n) \log^2 n)$ with constant error probability $\frac12$. (If the algorithm does not succeed, it returns \textsc{Fail}.)
\end{theorem}

We develop this algorithm in two steps: First, we prove our decomposition scheme for restricted graphs (\cref{sec:restricted-sssp:sec:decomposition}) and then we use the decomposition scheme to build an SSSP algorithm for restricted graphs (\cref{sec:restricted-sssp:sec:main-alg}).

\subsection{Decomposition for Restricted Graphs} \label{sec:restricted-sssp:sec:decomposition}
In this section, we prove the decomposition lemma:

\lemdecomposition*

\medskip
For the proof, we introduce some notation. Let $G_{\ge 0}$ denote the graph obtained by replacing negative edge weights by $0$ in the graph $G$.
A vertex $v$ is \emph{out-heavy} if $|\Bout_{G_{\ge 0}}(v, \frac\kappa4)| > \frac {n}2$ and \emph{out-light} if $|\Bout_{G_{\ge 0}}(v, \frac\kappa4)| \leq \frac{3n}4$. Note that there can be vertices which are both out-heavy and out-light. We similarly define \emph{in-light} and \emph{in-heavy} vertices with ``$\Bin_{G_{\ge 0}}$'' in place of ``$\Bout_{G_{\ge 0}}$''.

\begin{lemma}[Heavy-Light Classification] \label{lem:heavy-light}
There is an algorithm which, given a directed graph~$G$, labels every vertex correctly as either in-light or in-heavy (vertices which are both in-light and in-heavy may receive either label). The algorithm runs in time $\Order((m + n \log \log n) \log n)$ and succeeds with high probability.
\end{lemma}

Note that by applying this lemma to the graph $G^{rev}$ obtained by flipping the edge orientations, we can similarly classify vertices into out-light and out-heavy. We omit the proof for now as it follows easily from \cref{lem:approx-ball-size} which we state and prove in \cref{sec:ldd}.

We are ready to state the decomposition algorithm: First, label each vertex as out-light or out-heavy and as in-light or in-heavy using the previous lemma. Then, as long as $G$ contains a vertex $v$ which is labeled out-light or in-light (say it is out-light), we will carve out a ball around $v$. To this end, we sample a radius $r$ from the geometric distribution~$\Geom(20 \log n / \kappa)$, we cut the edges $\partial\Bout_{G_{\ge 0}}(v, r)$ (that is, the set of edges leaving $\Bout_{G_{\ge 0}}(v, r)$) and we remove all vertices in $\Bout_{G_{\ge 0}}(v, r)$ from the graph. We summarize the procedure in \cref{alg:decompose}. In what follows, we prove correctness of this algorithm.

\begin{algorithm}[t]
\caption{The graph decomposition. This algorithm $\textsc{Decompose}(G, \kappa)$ computes a subset of edges $S \subseteq E(G)$ satisfying the properties in \cref{lem:decomposition}.} \label{alg:decompose}
\smallskip
\begin{algorithmic}[1]
\Procedure{Decompose}{$G, \kappa$}
    \State Let $L^{in} \subseteq V(G)$ be the vertices labeled as in-light by \cref{lem:heavy-light} on $G$
    \State Let $L^{out} \subseteq V(G)$ be the vertices labeled as out-light by \cref{lem:heavy-light} on $G^{rev}$
    \State $S \gets \emptyset$
    \While{there is $v \in V(G) \cap L^{out}$} \label{alg:decompose:line:out-loop}
        \State Sample $r \sim \Geom(20 \log n / \kappa)$
        \State $S \gets S \cup \partial\Bout_{G_{\ge 0}}(v, r)$ \label{alg:decompose:line:out-cut}
        \State $G \gets G \setminus \Bout_{G_{\ge 0}}(v, r)$
    \EndWhile
    \smallskip
    \While{there is $v \in V(G) \cap L^{in}$} \label{alg:decompose:line:in-loop}
        \State Sample $r \sim \Geom(20 \log n / \kappa)$
        \State $S \gets S \cup \partial \Bin_{G_{\ge 0}}(v, r)$ \label{alg:decompose:line:in-cut}
        \State $G \gets G \setminus \Bin_{G_{\ge 0}}(v, r)$
    \EndWhile
    \State \Return $S$
\EndProcedure
\end{algorithmic}
\end{algorithm}

\begin{lemma}[Sparse Hitting of \cref{alg:decompose}] \label{lem:sparse-hitting}
Let $P$ be a shortest $s$-$v$-path in $G$ and let $S$ be the output of $\textsc{Decompose}(G, \kappa)$. Then $\Ex(|P \cap S|) \leq \Order(\log n)$.
\end{lemma}
\begin{proof}
Focus on any edge $e = (x, y) \in E(G)$. We analyze the probability that $e \in S$. We first analyze the probability of $e$ being included into $S$ in \cref{alg:decompose:line:out-cut} (and the same analysis applies to the case where the edge is included in \cref{alg:decompose:line:in-cut}). Focus on any iteration of the loop in \cref{alg:decompose:line:out-loop} for some out-light vertex $v$. There are three options:
\smallskip
\begin{itemize}
\item $x, y \not\in \Bout_{G_{\ge 0}}(v, r)$: The edge $e$ is not touched in this iteration. It might or might not be included in later iterations.
\item $x \in \Bout_{G_{\ge 0}}(v, r)$ and $y \not\in \Bout_{G_{\ge 0}}(v, r)$: The edge $e$ is contained in $\partial \Bout_{G_{\ge 0}}(v, r)$ and thus definitely included into $S$.
\item $y \in \Bout_{G_{\ge 0}}(v, r)$: The edge $e$ is definitely not included into $S$. Indeed, $e \not\in \partial \Bout_{G_{\ge 0}}(v, r)$, so we do not include $e$ into $S$ in this iteration. Moreover, as we remove $y$ from $G$ after this iteration, we will never consider the edge $e$ again.
\end{itemize}
\smallskip
Recall that the radius $r$ is sampled from the geometric distribution $\Geom(p)$ for $p := 20 \log n/\kappa$. Therefore, we have that
\begin{gather*}
    \Pr(e \in S) \leq \max_{v \in V} \; \Pr_{r \sim \Geom(p)}(y \not\in \Bout_{G_{\ge 0}}(v, r) \mid x \in \Bout_{G_{\ge 0}}(v, r)) \\
    \qquad= \max_{v \in V} \; \Pr_{r \sim \Geom(p)}(r < \dist_{G_{\ge 0}}(v, y) \mid r \geq \dist_{G_{\ge 0}}(v, x)) \\
    \qquad\leq \max_{v \in V} \; \Pr_{r \sim \Geom(p)}(r < \dist_{G_{\ge 0}}(v, x) + w_{G_{\ge 0}}(e) \mid r \geq \dist_{G_{\ge 0}}(v, x))
\intertext{By the memoryless property of geometric distributions, we may replace $r$ by the (nonnegative) random variable $r' := r-\dist_{G_{\ge 0}}(v, x)$:}
    \qquad= \max_{v \in V} \; \Pr_{r' \sim \Geom(p)}(r' < w_{G_{\ge 0}}(e)) \\
    \qquad= \Pr_{r' \sim \Geom(p)}(r' < w_{G_{\ge 0}}(e)) \\
    \qquad\leq p \cdot w_{G_{\ge 0}}(e).
\end{gather*}
The last inequality follows since we can interpret $r' \sim \Geom(p)$ as the number of coin tosses until we obtain heads, where each toss is independent and lands heads with probability $p$. Thus, by a union bound, $\Pr(r' < w_{G_{\ge 0}}(e))$ is upper bounded by the probability that at least one of $w_{G_{\ge 0}}(e)$ coin tosses lands heads.

Now consider a shortest $s$-$v$-path $P$ in $G$. Recall that $w_G(P) \leq 0$, since $G$ is a restricted graph. Hence, $P$ contains at most $\kappa(G) \leq \kappa$ edges with negative weight (i.e., with weight exactly $-1$). It follows that $w_{G_{\ge 0}}(P) \leq \kappa$ and thus finally:
\begin{equation*}
    \Ex(|P \cap S|) = \sum_{e \in P} \Pr(e \in S) = \sum_{e \in P} p \cdot w_{G_{\ge 0}}(e) \leq p \cdot w_{G_{\ge 0}}(P) \leq p \kappa = \Order(\log n). \qedhere
\end{equation*}
\end{proof}

In what follows, we will need the following lemma.

\begin{lemma} \label{lem:closed-walk}
Let $G$ be a directed graph. Then $\min_{C} \bar w(C) = \min_Z \bar w(Z)$ where $C$ ranges over all cycles and $Z$ ranges over all \emph{closed walks} in $G$.
\end{lemma}
\begin{proof}
Write $c = \min_C \bar w(C)$ and $z = \min_Z \bar w(Z)$. It suffices to prove that $c \leq z$. Take the closed walk $Z$ witnessing $z$ with the minimum number of edges. If $Z$ is a cycle, then we clearly have $c \leq z$. Otherwise, $Z$ must revisit at least one vertex and can therefore be split into two closed walks $Z_1, Z_2$. By the minimality of $Z$ we have $\bar w(Z_1), \bar w(Z_2) > z$. But note that
\begin{equation*}
    z \cdot |Z| = w(Z) = w(Z_1) + w(Z_2) > z \cdot |Z_1| + z \cdot |Z_2| = z \cdot |Z|,
\end{equation*}
a contradiction.
\end{proof}

\begin{lemma}[Progress of \cref{alg:decompose}] \label{lem:progress}
Let $S$ be the output of $\textsc{Decompose}(G, \kappa)$. Then, with high probability, any strongly connected component $C$ in $G \setminus S$ satisfies (i) $|C| \leq \frac34 |V(G)|$ or (ii)~$\kappa(G[C]) \leq \frac\kappa2$.
\end{lemma}
\begin{proof}
Throughout, condition on the event that the heavy-light classification was successful (which happens with high probability). Observe that whenever we carve out a ball $\Bout_{G_{\ge 0}}(v, r)$ and include its outgoing edges~$\partial\Bout_{G_{\ge 0}}(v, r)$ into $S$, then any two vertices $x \in \Bout_{G_{\ge 0}}(v, r)$ and $y \not\in \Bout_{G_{\ge 0}}(v, r)$ cannot be part of the same strongly connected component in $G \setminus S$ (as there is no path from $x$ to $y$). The same argument applies to $\Bin_{G_{\ge 0}}(v, r)$.

Therefore, there are only two types of strongly connected components: (i) Those contained in $\Bout_{G_{\ge 0}}(v, r)$ or $\Bin_{G_{\ge 0}}(v, r)$, and (ii) those in the remaining graph after it no longer contains light vertices. We argue that each component of type (i) satisfies that $|C| \leq \frac34 |V(G)|$ (with high probability) and that each component of type (ii) satisfies $\kappa(G[C]) \leq \frac\kappa2$.

In case (i) we have $|C| \leq |\Bout_{G_{\ge 0}}(v, r)|$. Since $v$ is out-light, it follows that $|C| \leq \frac34 |V(G)|$ whenever $r \leq \frac\kappa4$. This event happens with high probability as:
\begin{equation*}
    \Pr_{r \sim \Geom(20 \log n / \kappa)}\left(r > \frac\kappa4\right)
        \leq \left(1 - \frac{20 \log n}{\kappa}\right)^{\frac\kappa4}
        \leq \exp(-5 \log n) \leq n^{-5}.
\end{equation*}
The number of iterations is bounded by $n$, thus by a union bound we never have $r > \frac\kappa4$ with probability at least $1 - n^{-4}$.
A similar argument applies if we carve $\Bin_{G_{\ge 0}}(v, r)$ when $v$ is in-light.

Next, focus on case (ii). Let $C$ be a strongly connected component in the remaining graph~$G$ after carving out all balls centered at light vertices. Suppose that $\kappa(G[C]) > \frac\kappa2$. We will construct a closed walk $Z$ in $G$ with mean weight $\bar w(Z) < 1$, contradicting the assumption that $G$ is restricted by \cref{lem:closed-walk}. Let~$P$ be the $s$-$v$-path in $G[C \cup \set s]$ of nonpositive weight witnessing the largest number of negative edges (i.e., the path that witnesses $\kappa(G[C \cup \set s])$), and let $u$ be the first vertex (after $s$) on that path $P$. Let~$P_1$ be the $u$-$v$-path obtained by removing the $s$-$u$-edge from $P$. Since the $s$-$u$-edge has weight~$0$, we have that $w(P_1) \leq 0$ and that~$P_1$ contains more than $\frac\kappa2$ negative-weight edges. Since $u, v$ are both out-heavy and in-heavy vertices in the original graph $G$, we have that \makebox{$|\Bout_{G_{\ge 0}}(v, \frac\kappa4)|, |\Bin_{G_{\ge 0}}(u, \frac\kappa4)| > \frac n2$}. It follows that these two balls must intersect and that there exists a $v$-$u$-path $P_2$ of length~\makebox{$w(P_2) \leq \frac\kappa4 + \frac{\kappa}{4} = \frac{\kappa}{2}$}. Combining $P_1$ and $P_2$, we obtain a closed walk $Z$ with total weight $w(Z) \leq \frac\kappa2$ containing more than $\frac\kappa2$ (negative-weight) edges. It follows that $\bar w(Z) < 1$ yielding the claimed contradiction.
\end{proof}

\begin{proof}[Proof of \cref{lem:decomposition}]
The correctness is immediate by the previous lemmas: \cref{lem:progress} proves the progress property, and \cref{lem:sparse-hitting} the sparse hitting property. Next, we analyze the running time. Computing the heavy-light classification takes time $\Order((m + n \log \log n) \log n)$ due to~\cref{lem:heavy-light}. Sampling each radius~$r$ from the geometric distribution $\Geom(20 \log n / \kappa)$ runs in expected constant time in the word RAM with word size $\Omega(\log n)$~\cite{BringmannF13}, so the overhead for sampling the radii is $\Order(n)$ in expectation. To compute the balls we use Dijkstra's algorithm. Using Thorup's priority queue~\cite{Thorup03}, each vertex explored in Dijkstra's takes time $\Order(\log\log n)$ and each edge time $\Order(1)$. Since every vertex contained in some ball is removed from subsequent iterations, a vertex participates in at most one ball. Note that a naive implementation of this would reinitialize the priority queue and distance array at each iteration of the while-loop. To avoid this, we initialize the priority queue and array of distances once, before the execution of the while-loops. Then, at the end of an iteration of the while-loop we reinitialize them in time proportional to the removed vertices and edges (this is the same approach as in the BNW algorithm~\cite{BernsteinNW22}). Thus, the overall time to compute all the balls is indeed $\Order(m + n\log\log n)$.
\end{proof}

\subsection{Proof of \texorpdfstring{\cref{thm:restricted-sssp}}{Theorem \ref{thm:restricted-sssp}}} \label{sec:restricted-sssp:sec:main-alg}
With the graph decomposition in hand, we can present our full algorithm for Restricted SSSP. The overall structure closely follows the BNW algorithm (see~\cite[Algorithm 1]{BernsteinNW22}).

We start with the following crucial definition.

\begin{definition} \label{def:eta}
    Let $G$ be a directed graph with a designated source vertex $s$.
    For any vertex $v \in V(G)$, we denote by $\eta_G(v)$ the smallest number of \emph{negative-weight} edges in any shortest $s$-$v$-path.
\end{definition}

The next proposition captures the relationship between the parameters $\kappa(G)$ and $\eta_G(\cdot)$ when $G$ is restricted (see~\cref{def:kappa,def:eta}).
\begin{proposition}\label{prop:bound-eta}
    Let $G$ be a restricted graph with source vertex $s$.
    Then, for every vertex $v \in V$ it holds that $\eta_G(v) \leq \kappa(G)$.
\end{proposition}
\begin{proof}
	Fix a vertex $v$. Let $P$ be a shortest $s$-$v$ path witnessing $\eta_G(v)$ (see~\cref{def:eta}). Since $G$ is restricted, it does not contain negative cycles and thus $P$ is a simple path.
    Furthermore, since there is an edge from $s$ to $v$ of weight 0, it follows that $w_G(P) \leq 0$.
    Recall that $\kappa(G)$ is the maximum number of negative edges in any simple path which starts at $s$ and has nonpositive weight (see~\cref{def:kappa}).
    Therefore, it follows that $\eta_G(v) \leq \kappa(G)$.
\end{proof}

Next, we use two lemmas from~\cite{BernsteinNW22}:

\begin{restatable}[Dijkstra with Negative Weights, similar to {{{\cite[Lemma~3.3]{BernsteinNW22}}}}]{lemma}{lemlazydijkstra} \label{lem:lazy-dijkstra}
Let $G$ be a directed graph with source vertex $s \in V(G)$ that does not contain a negative cycle. There is an algorithm that computes a shortest path tree from $s$ in time $\Order(\sum_v (\deg(v) + \log \log n) \cdot \eta_G(v))$. (If $G$ contains a negative cycle, the algorithm does not terminate.)
\end{restatable}

The main differences to \cite[Lemma~3.3]{BernsteinNW22} are that we use a faster priority queue for Dijkstra and that \cite[Lemma~3.3]{BernsteinNW22} is restricted to graphs of constant maximum degree. Therefore, we devote \cref{sec:dijkstra} to a self-contained proof of \cref{lem:lazy-dijkstra}.

\begin{lemma}[DAG Edges, {{{\cite[Lemma~3.2]{BernsteinNW22}}}}] \label{lem:dag-edges}
Let $G$ be a directed graph with nonnegative edge weights inside its SCCs. Then we can compute a potential function $\phi$ such that $G_\phi$ has nonnegative edge weights (everywhere) in time $\Order(n + m)$.
\end{lemma}
\begin{proof}[Proof Sketch]
For the complete proof, see~\cite[Lemma~3.2]{BernsteinNW22}. The idea is to treat the graph as a DAG of SCCs, and to assign a potential function $\phi$ to every SCC such that the DAG edges become nonnegative. One way to achieve this is by computing a topological ordering, and by assigning $\phi(v)$ to be $W$ times the rank of $v$'s SCC in that ordering (here, $-W$ is the smallest weight in $G$). Then $G_\phi$ satisfies the claim.
\end{proof}

\subparagraph{The Algorithm}
We are ready to state the algorithm; see \cref{alg:restricted-sssp} for the pseudocode. Recall that $\kappa(G)$ is the maximum
number of negative edges in any path $P$ starting at $s$ with $w(P) \leq 0$ (\cref{def:kappa}).
If $\kappa(G) \leq 2$, we run \cref{lem:lazy-dijkstra} to compute the distances from $s$. Otherwise, we start with applying our graph decomposition. That is, we compute a set of edges $S$, such that any strongly connected component $C$ in the graph $G \setminus S$ is either small or has an improved $\kappa$-value. This constitutes enough progress to solve the induced graphs~$G[C \cup \set s]$ recursively. The recursive calls produce shortest path trees and thereby a potential function $\phi_1$ such that $G_{\phi_1}$ has nonnegative edge weights inside each SCC. We then add back the missing edges by first calling \cref{lem:dag-edges} (to fix the edges $e \not\in S$ between strongly connected components) and then \cref{lem:lazy-dijkstra} (to fix the edges $e \in S$). The correctness proof is easy:

\begin{algorithm}[t]
\caption{Solves the negative-weight SSSP problem on restricted graphs. The procedure $\textsc{RestrictedSSSP}(G, \kappa)$ takes a restricted graph $G$ and an upper bound $\kappa \geq \kappa(G)$, and computes a shortest path tree from the designated source vertex $s$.} \label{alg:restricted-sssp}
\smallskip
\begin{algorithmic}[1]
\Procedure{RestrictedSSSP}{$G, \kappa$}
    \If{$\kappa \leq 2$}
        \State Run \cref{lem:lazy-dijkstra} on $G$ from $s$ and \Return the computed shortest path tree \label{alg:restricted-sssp:line:base-case}
    \EndIf
    \State Compute $S \gets \textsc{Decompose}(G, \kappa)$ (see \cref{lem:decomposition}) \label{alg:restricted-sssp:line:decompose}
    \State Compute the strongly connected components $C_1, \dots, C_\ell$ of $G \setminus S$ (see \cref{lem:scc}) \label{alg:restricted-sssp:line:scc}
    \For{$i \gets 1, \dots, \ell$}
        \IfThenElse{$|C_i| \leq \frac{3n}4$}{$\kappa_i \gets \kappa$}{$\kappa_i \gets \frac\kappa2$} \label{alg:restricted-sssp:line:condition}
        \State Recursively call $\textsc{RestrictedSSSP}(G[C_i \cup \set{s}], \kappa_i)$ \label{alg:restricted-sssp:line:recur}
        \State Let $\phi_1(v) = \dist_{G[C_i \cup \set{s}]}(s, v)$ for all $v \in C_i$
    \EndFor
    \State Run \cref{lem:dag-edges} on $(G \setminus S)_{\phi_1}$ to obtain a potential $\phi_2$ \label{alg:restricted-sssp:line:dag-edges}
    \State Run \cref{lem:lazy-dijkstra} on $G_{\phi_2}$ and \Return the computed shortest path tree \label{alg:restricted-sssp:line:lazy-dijkstra}
\EndProcedure
\end{algorithmic}
\end{algorithm}

\begin{lemma}[Correctness of \cref{alg:restricted-sssp}] \label{lem:restricted-sssp-correctness}
Let $G$ be an arbitrary directed graph (not necessarily restricted), and let $\kappa$ be arbitrary. Then, if $\RestrictedSSSP(G, \kappa)$ terminates, it correctly computes a shortest path tree from the designated source vertex $s$.
\end{lemma}
\begin{proof}
If $\kappa \leq 2$ and the call in \cref{alg:restricted-sssp:line:base-case} terminates, then it correctly computes a shortest path tree due to \cref{lem:lazy-dijkstra}. If $\kappa > 2$, then in \cref{alg:restricted-sssp:line:dag-edges} we compute a potential function~$\phi_2$ and in \cref{alg:restricted-sssp:line:lazy-dijkstra} we run \cref{lem:lazy-dijkstra} to compute a shortest path tree in the graph $G_{\phi_2}$. Assuming that \cref{lem:lazy-dijkstra} terminates, this computation is correct since $G_{\phi_2}$ is equivalent to $G$.
\end{proof}

\begin{lemma}[Running Time of \cref{alg:restricted-sssp}] \label{lem:restricted-sssp-time}
Let $G$ be a restricted graph with $\kappa(G) \leq \kappa$. Then $\textsc{RestrictedSSSP}(G, \kappa)$ runs in expected time $\Order((m + n \log \log n) \log^2 n)$.
\end{lemma}
\begin{proof}

We first analyze the running time of a single call to \cref{alg:restricted-sssp}, ignoring the time spent in recursive calls.
For the base case, when $\kappa(G) \leq 2$, the running time of \cref{alg:restricted-sssp:line:base-case} is $\Order(m + n \log \log n)$ by~\cref{lem:lazy-dijkstra,prop:bound-eta}.
Otherwise, the call to $\textsc{Decompose}(G, \kappa)$ in \cref{alg:restricted-sssp:line:decompose} runs in time $\Order((m + n \log \log n) \log n)$ by \cref{lem:decomposition}. Computing the strongly connected components in $G \setminus S$ is in linear time $\Order(m + n)$, and so is the call to \cref{lem:dag-edges} in \cref{alg:restricted-sssp:line:dag-edges}.

Analyzing the running time of \cref{alg:restricted-sssp:line:lazy-dijkstra} takes some more effort.
Recall that $\eta_{G_{\phi_2}}(v)$ is the minimum number of negative edges in any $s$-$v$ path in $G_{\phi_2}$ (see~\cref{def:eta}).
Our intermediate goal is to bound $\Ex(\eta_{G_{\phi_2}}(v)) = \Order(\log n)$ for all vertices $v$.
Let $S$ be the set of edges computed by the decomposition, as in the algorithm. We proceed in three steps:
\smallskip

\begin{itemize}
    \item\emph{Claim 1: $G_{\phi_1} \setminus S$ has nonnegative edges inside its SCCs.} The recursive calls in \cref{alg:restricted-sssp:line:recur} correctly compute the distances by \cref{lem:restricted-sssp-correctness}. Hence, for any two nodes $u, v \in C_i$, we have that~$w_{\phi_1}(u, v) = w(u, v) + \dist_{G[C_i \cup \set{s}]}(s, u) - \dist_{G[C_i \cup \set{s}]}(s, v) \geq 0$, by the triangle inequality.
    \item\emph{Claim 2: $G_{\phi_2} \setminus S$ has only nonnegative edges.} This is immediate by \cref{lem:dag-edges}.
    \item\emph{Claim 3: For every node $v$ we have $\Ex(\eta_{G_{\phi_2}}(v)) \leq \Order(\log n)$.} Let $P$ be a shortest $s$\=/$v$\=/path in $G$. Since $G$ and $G_{\phi_2}$ are equivalent,~$P$~is also a shortest path in $G_{\phi_2}$. By the previous claim, the only candidate negative edges in~$P$ are the edges in $S$. Therefore, we have that $\Ex(\eta_{G_{\phi_2}}(v)) \leq \Ex(|P \cap S|) = \Order(\log n)$, by \cref{lem:decomposition}.
\end{itemize}
\smallskip
The expected running time of \cref{alg:restricted-sssp:line:lazy-dijkstra} is thus bounded by
\begin{gather*}
    \Order\left(\sum_{v \in V(G)} (\deg(v) + \log \log n) \cdot \Ex(\eta_{G_{\phi_2}}(v))\right) \\
    \qquad= \Order\left(\sum_{v \in V(G)} (\deg(v) + \log \log n) \cdot \log n\right)  \\
    \qquad= \Order((m + n \log \log n) \log n).
\end{gather*}
Therefore, a single execution of \cref{alg:restricted-sssp} runs in time $\Order((m + n \log \log n) \log n)$; let $c$ denote the hidden constant in the $O$-notation.

We finally analyze the total running time, taking into account the recursive calls. We inductively prove that the running time is bounded by $c(m + n \log\log n) \log n \cdot \log_{4/3}(n \kappa)$.

We claim that for each recursive call on a subgraph $G[C_i \cup \set s]$, where $C_i$ is a strongly connected component in~\makebox{$G \setminus S$}, it holds that (i) $G[C_i \cup \set s]$~is a restricted graph and that (ii)~\makebox{$\kappa(G[C_i \cup \set s]) \leq \kappa_i$}. To see (i), observe that any subgraph of $G$ containing $s$ is also restricted. To show (ii), we distinguish two cases: Either $|C_i| \leq \frac{3n}4$, in which case we trivially have $\kappa(G[C_i \cup \set{s}]) \leq \kappa(G) \leq \kappa = \kappa_i$. Or $|C_i| > \frac{3n}4$, and in this case \cref{lem:decomposition} guarantees that $\kappa(G[C_i \cup \set{s}]) \leq \frac\kappa2 = \kappa_i$. It follows by induction that each recursive call runs in time~$c \cdot (|E(G[C_i \cup \set{s}])| + |C_i| \log\log n) \log n \cdot \log_{4/3}(|C_i| \kappa_i)$. Moreover, observe that in either case we have $|C_i| \kappa_i \leq \frac34 n \kappa$. Therefore the total time can be bounded by
\begin{gather*}
    c (m + n \log\log n) \log n
        + \sum_{i=1}^\ell c \cdot (|E(G[C_i \cup \set{s}])| + |C_i| \log\log n) \log n \cdot \log_{4/3}(|C_i| \kappa_i) \\
    \qquad\leq c (m + n \log\log n) \log n
        \\
        \qquad\qquad + \sum_{i=1}^\ell c \cdot (|E(G[C_i \cup \set{s}])| + |C_i| \log\log n) \log n \cdot (\log_{4/3}(n \kappa) - 1) \\
    \qquad\leq c (m + n \log\log n) \log n + c(m + n \log\log n) \log n \cdot (\log_{4/3}(n \kappa) - 1) \\
    \qquad= c m \log n \log\log n \cdot \log_{4/3}(n \kappa),
\end{gather*}
where in the third step we used that $\sum_i |E(G[C_i \cup \set{s}])| \leq m$ and that $\sum_i |C_i| \leq n$. This completes the running time analysis.
\end{proof}

\begin{proof}[Proof of \cref{thm:restricted-sssp}]
This proof is almost immediate from the previous two \cref{lem:restricted-sssp-correctness,lem:restricted-sssp-time}. In combination, these lemmas prove that \cref{alg:restricted-sssp} is a Las Vegas algorithm for the Restricted SSSP problem which runs in expected time $\Order((m + n \log \log n) \log^2 n)$. By interrupting the algorithm after twice its expected running time (and returning \textsc{Fail} in that case), we obtain a Monte Carlo algorithm with worst-case running time $\Order((m + n \log \log n) \log^2 n)$ and error probability $\frac12$ as claimed.
\end{proof}

We remark that \cref{alg:restricted-sssp} is correct even if the input graph $G$ is not restricted---therefore, whenever $G$ contains a negative cycle, the algorithm cannot terminate.

\section{SSSP on Graphs without Negative Cycles} \label{sec:sssp-without-neg-cycles}

In this section we present the $\Order((m + n \log \log n) \log^2(n) \log(nW))$-time algorithm for SSSP on graphs~$G$ without negative cycles. Later in \cref{sec:neg-cycle}, we will remove the assumption that~$G$ does not contain negative cycles, and strengthen the algorithm to find a negative cycle if it exists.

The main idea is to use \emph{scaling} and some tricks for probability amplification in order to extend our algorithm for restricted graphs developed in~\cref{sec:restricted-sssp}. More precisely, we use the standard \emph{scaling technique}~\cite{Gabow83,GabowT89,Goldberg95,BernsteinNW22} to reduce the computation of SSSP in an arbitrary graph (without negative cycles) to the case of restricted graphs. Formally, we prove the following theorem:

\begin{theorem}[Scaling Algorithm for SSSP] \label{thm:sssp}
There is a Las Vegas algorithm which, given a directed graph $G$ without negative cycles and with a source vertex $s \in V(G)$, computes a shortest path tree from~$s$, running in time \makebox{$\Order(T_{\TRestrictedSSSP}(m, n) \cdot \log(nW))$} with~high~probability (and in expectation).
\end{theorem}

\subparagraph{One-Step Scaling}
The idea of the scaling algorithm is to increase the smallest weight in~$G$ step-by-step, while maintaining an equivalent graph. The following \cref{lem:scale} gives the implementation of one such scaling step as a direct reduction to Restricted SSSP.

\begin{lemma}[One-Step Scaling] \label{lem:scale}
Let $G$ be a directed graph that does not contain a negative cycle and with minimum weight greater than $-3W$ (for some integer $W \geq 1$). There is an algorithm $\textsc{Scale}(G)$ computing $\phi$ such that $G_\phi$ has minimum weight greater than~$-2W$, which succeeds with constant probability (if the algorithm does not suceed, it returns \textsc{Fail}) and runs in time~$\Order(T_{\TRestrictedSSSP}(m, n))$.
\end{lemma}
\begin{proof}
We construct a restricted graph $H$ as a copy of $G$ with modified edge weights $w_H(e) = \ceil{w_G(e) / W} + 1$. We also add a source vertex $s$ to $H$, and put edges of weight $0$ from~$s$ to all other vertices. We compute a shortest path tree from~$s$ in $H$ using~\cref{thm:restricted-sssp}, and return the potential $\phi$ defined by $\phi(v) = W \cdot \dist_H(s, v)$. For the pseudocode, see \cref{alg:scale}. Note that the running time is dominated by computing shortest paths in a restricted graph.

\begin{algorithm}[t]
\caption{One step of the scaling algorithm. Given a graph $G$ with minimum weight greater than~$-3W$, $\textsc{Scale}(G)$ computes a potential~$\phi$ such that $G_\phi$ has minimum weight greater than $-2W$. See \cref{lem:scale}.} \label{alg:scale}
\begin{algorithmic}[1]
\Procedure{Scale}{$G$}
    \State Let $W$ be such that all weights in $G$ are greater than $-3W$
    \State Let $H$ be a copy of $G$ with edge weights $w_H(e) = \ceil{w_G(e) / W} + 1$, and add
    \Statex[2] an artificial source vertex $s$ to $H$ with weight-$0$ edges to all other vertices
    \State Compute a shortest path tree from $s$ in the restricted graph $H$ using~\cref{thm:restricted-sssp}\label{alg:scale:line:sssp}
    \State Let $\phi$ be the potential defined by $\phi(v) = W \cdot \dist_H(s, v)$
    \State \Return $\phi$
\EndProcedure
\end{algorithmic}
\end{algorithm}

To prove that the algorithm is correct, we first check that $H$ is indeed restricted (see \cref{def:restricted}):
\smallskip
\begin{itemize}
    \item Each edge weight satisfies $w_H(e) = \ceil{w_G(e) / W} + 1 \geq \ceil{(-3 W + 1) / W} + 1 = -1$.
    \item Consider any cycle $C$ in $H$. Recall that $w_G(C) \geq 0$ (as $G$ does not contain negative cycles), and thus
    \begin{equation*}
        \bar w_H(C) = \frac{w_H(C)}{|C|} = \frac1{|C|} \sum_{e \in C} w_H(e) = 1 + \frac1{|C|} \sum_{e \in C} \ceil*{w_G(e) \cdot \frac{1}{W}} \geq 1 + \frac{w_G(C)}{W|C|} \geq 1.
    \end{equation*}
    In particular, the minimum cycle mean in $H$ is at least $1$.
    \item Finally, we have artificially added a source vertex $s$ to $H$ with weight-$0$ edges to all other vertices.
\end{itemize}
\smallskip

It remains to prove that the potential $\phi$ defined by $\phi(v) = W \cdot \dist_H(s, v)$ satisfies that~$G_\phi$ has minimum edge weight more than $-2W$. Consider any edge $e = (u, v)$. Since by definition~$w_H(e) < w_G(e) \cdot \frac{1}{W} + 2$, we have that $w_G(e) > W \cdot (w_H(e) - 2)$. It follows that
\begin{align*}
    w_{G_\phi}(e)
    &= w_G(e) + \phi(u) - \phi(v) \\
    &= w_G(e) + W \cdot \dist_H(s, u) - W \cdot \dist_H(s, v) \\
    &> -2W + W \cdot w_H(e) + W \cdot \dist_H(s, u) - W \cdot \dist_H(s, v) \\
    &\geq -2W.
\end{align*}
In the last step we have used the triangle inequality $\dist_H(s, v) \leq \dist_H(s, u) + w_H(u, v)$.

Finally, we argue that the algorithm succeeds with constant probability. Observe that the algorithm succeeds if the computation of the shortest path tree from $s$ succeeds in~\cref{alg:scale:line:sssp} (indeed, all other steps are deterministic). Since $H$ is restricted, \cref{thm:restricted-sssp} guarantees that this holds with constant probability, and if it does not suceed it returns \textsc{Fail}, completing the proof.
\end{proof}

\subparagraph{The Complete Scaling Algorithm}
We are ready to state the algorithm $\textsc{SSSP}(G, s)$ which implements \cref{thm:sssp}. We construct a graph $G_0$ by multiplying every edge weight of $G$ by~$4n$. Then, for $i = 0, \dots, L - 1$ where $L = \Theta(\log(nW))$, we call $\textsc{Scale}(G_i)$ (we repeat the call until it succeeds) to obtain a potential $\phi_i$ and set $G_{i+1} := (G_i)_{\phi_i}$. Next, we construct a graph~$G^*$ as a copy of $G_L$, with every negative edge weight replaced by 0. Finally, we compute a shortest path tree in $G^*$ using Dijkstra's algorithm. For the details, see the pseudocode in \cref{alg:sssp}.

\begin{algorithm}[t]
\caption{The fast SSSP algorithm. In a given graph $G$ without negative cycles, it computes a shortest path tree in $G$ from the given source vertex $s$.} \label{alg:sssp}
\begin{algorithmic}[1]
\Procedure{SSSP}{$G, s$}
    \State Let $-W$ be the smallest edge weight in $G$
    \State Let $G_0$ be a copy of $G$ with edge weights $w_{G_0}(e) = 4n \cdot w_G(e)$
    \For{$i \gets 0, \dots, L - 1$ where $L = \Theta(\log(nW))$} \label{alg:sssp:line:for-loop}
        \State $\phi_i \gets \textsc{Scale}(G_i)$ (rerun the algorithm until it succeeds)
        \State $G_{i+1} \gets (G_i)_{\phi_i}$
    \EndFor
    \State Let $G^*$ be a copy of $G_L$ with negative weights replaced by $0$
    \State \Return a shortest path tree in $G^*$ from $s$, computed by Dijkstra's algorithm \label{alg:sssp:line:dijkstra}
\EndProcedure
\end{algorithmic}
\end{algorithm}

\begin{lemma}[Running Time of \cref{alg:sssp}] \label{thm:sssp-time}
If $G$ does not contain a negative cycle, then $\textsc{SSSP}(G, s)$ runs in time $\Order(T_{\TRestrictedSSSP}(m, n) \cdot \log(nW))$ with high probability (and in expectation).
\end{lemma}
\begin{proof}
We analyze the running time of the for-loop, which runs for $L = \Order(\log(nW))$ iterations.
Each iteration repeatedly calls $\textsc{Scale}(G_i)$ until one such call succeeds. By~\cref{lem:scale}, a single call succeeds with constant probability (say, $\frac12$) and runs in time $\Order(T_{\TRestrictedSSSP}(m, n))$. We can therefore model the running time of the $i$-th iteration by $\Order(X_i \cdot T_{\TRestrictedSSSP}(m, n))$ where~$X_i \sim \Geom(\frac12)$ is a geometric random variable. Therefore, by Chernoff's bound, the time of the for-loop is bounded by $\Order(\sum_{i=0}^{L-1} X_i \cdot T_{\TRestrictedSSSP}(m, n)) = \Order(T_{\TRestrictedSSSP}(m, n) \cdot L)$ with probability at least $1 - \exp(-\Omega(L)) \geq 1 - n^{-\Omega(1)}$.
Finally, observe that $T_{\TRestrictedSSSP}(m,n) = \Omega(m + n)$, and therefore the call to Dijkstra's algorithm in~\cref{alg:sssp:line:dijkstra} is dominated by the time spent in the for-loop.
\end{proof}

\begin{lemma}[Correctness of \cref{alg:sssp}] \label{thm:sssp-correctness}
If $G$ does not contain a negative cycle, then \cref{alg:sssp} correctly computes a shortest path tree from $s$.
\end{lemma}
\begin{proof}
Consider an execution of \cref{alg:sssp}. We prove that any shortest path in $G^*$ is a shortest path in $G$, and hence
the shortest path tree from $s$ computed in $G^*$ is also a shortest path tree from $s$ in $G$, implying correctness. We proceed in three steps:
\smallskip
\begin{itemize}
\item As $G_0$ is a copy of $G$ with scaled edge weights $w_{G_0}(e) = 4n \cdot w_G(e)$, any path $P$ also has scaled weight $w_{G_0}(P) = 4n \cdot w_G(P)$ and therefore $G$ and $G_0$ are equivalent.
\item Since the graphs $G_0, \dots, G_L$ are obtained from each other by adding potential functions, they are equivalent (see \cref{lem:johnson}). Moreover, by the properties of \cref{lem:scale}, the smallest weight~$-W$ increases by a factor~$\frac{2}{3}$ in every step until $G_L$ has smallest weight at least $-3$. Here we use that $L = \Omega(\log(nW))$ for sufficiently large hidden constant.
\item $G^*$ is the graph obtained from $G_L$ by replacing negative-weight edges by $0$-weight edges.
Consider any non-shortest $u$-$v$-path $P'$ in $G_L$. We will show that $P'$ is also not a shortest $u$-$v$ path in $G^*$, which completes
the argument. Towards that end, let $P$ be any shortest $u$-$v$-path. Recall that~$G_L$ equals $(G_0)_\phi$ for some potential function $\phi$. Therefore:
\begin{gather*}
    w_{G_L}(P') - w_{G_L}(P) \\
    \qquad= w_{G_0}(P') + \phi(u) - \phi(v) - w_{G_0}(P) - \phi(u) + \phi(v) \\
    \qquad= w_{G_0}(P') - w_{G_0}(P) \\
    \qquad\geq 4n,
\end{gather*}
where the last inequality uses that the weights of $P$ and $P'$ in $G_0$ differ by at least $4n$ (this is why we scaled the edge weights by $4n$ in $G_0$). Finally, recall that by transitioning to $G^*$ we can increase the weight of any path by at most $3 \cdot (n - 1)$. It follows that
\[
    w_{G^*}(P') - w_{G^*}(P) \geq w_{G_L}(P') - w_{G_L}(P) - 3 \cdot (n - 1) \geq 4n - 3\cdot(n-1) >  0,
\]
and therefore, $P'$ is not a shortest $u$-$v$-path in $G^*$.
Hence, a shortest path in $G^*$ is also a shortest path in $G_L$, and since $G_L$ is equivalent to $G$, it is also a shortest path in $G$. \qedhere
\end{itemize}
\end{proof}

The proof of \cref{thm:sssp} is immediate by combining \cref{thm:sssp-time,thm:sssp-correctness}.

We end this section with the following lemma, which will be useful in the next section.
\begin{lemma}\label{lem:sssp-termination}
Let $G$ be a directed weighted graph and $s \in V(G)$. If $\textsc{SSSP}(G,s)$ terminates, then $G$ does not contain negative cycles.
\end{lemma}
\begin{proof}
    Assume for the sake of contradiction that $G$ has a negative cycle $C$ and that $\textsc{SSSP}(G, s)$ terminates.
    Consider the graph $G_L$ which is constructed in the last iteration of the for-loop in~\cref{alg:sssp:line:for-loop}.
    Note that $G_L$ is equivalent to $G_0$, since it was obtained by adding potential functions.
    Observe that the weight of $C$ in $G_0$ and $G_L$ is at most $-4n$, since it was negative in $G$ and we scaled by a factor $4n$ (see~\cref{lem:johnson}).
    Recall that we chose $L = \Theta(\log(nW))$ with large enough hidden constant so that the smallest weight in $G_L$ is at least $-3$.
    This implies that the weight of the minimum cycle in $G_L$ is at least $-3n$, a contradiction.
\end{proof}

\section{Finding Negative Cycles} \label{sec:neg-cycle}

In \cref{sec:sssp-without-neg-cycles} we developed an algorithm to compute a shortest path tree with high probability in a graph without negative cycles. In this section, we extend that result to \emph{find} a negative cycle if it exists. As a warm-up, we observe that the \textsc{SSSP} algorithm developed in~\cref{thm:sssp} can be used to \emph{detect} the presence of
a negative cycle with high probability:

\begin{corollary}\label{cor:detect-negcycle}
    Let $G$ be a directed graph. Then, there is an algorithm
    $\textsc{DetectNegCycle}(G)$ with the following properties:
    \begin{itemize}
        \item If $G$ has a negative cycle, then the algorithm reports \textsc{NegCycle}.
        \item If $G$ does not have a negative cycle, then with high probability it returns \textsc{NoNegCycle}
        \item It runs in time $\Order(T_{\TRestrictedSSSP}(m, n) \log(nW))$.
    \end{itemize}
\end{corollary}
\begin{proof}
    The algorithm adds a dummy source $s$ connected with 0-weight edges to all vertices in $G$ and runs
    $\textsc{SSSP}(G, s)$. If it finishes within its time budget, we return \textsc{NoNegCycle}, otherwise
    we interrupt the computation and return \textsc{NegCycle}. The running time follows immediately by the guarantee
    of~\cref{thm:sssp}.

    Now we argue about correctness.
    If $G$ contains no negative cycles, then the algorithm returns \textsc{NoNegCycle} with high probability due to~\cref{thm:sssp}.
    If $G$ contains a negative cycle, then \cref{lem:sssp-termination} implies that $\text{SSSP}(G, s)$ does not terminate, so in this
    case we always report \textsc{NegCycle}.
\end{proof}

\emph{Finding} the negative cycle though, requires some more work. Towards this end, we follow the ideas of~\cite{BernsteinNW22}. They reduced the problem of finding a negative cycle to a problem called \textsc{Threshold}, which we define next. We will use the following notation: given a directed graph $G$ and an integer $M$, we write $G^{+M}$ to denote
the graph obtained by adding $M$ to every edge weight of $G$.

\begin{definition}[Threshold]\label{def:threshold}
    Given a directed graph $G$, $\textsc{Threshold}(G)$ is the smallest integer $M^* \geq 0$ such that
    $G^{+M^*}$ contains no negative cycle.
\end{definition}
For a graph $G$, we write $T_{\textsc{Threshold}}(m, n)$ for the optimal running time of an algorithm
computing $\textsc{Threshold}(G)$ with high probability.

\medskip

The remainder of the section is organized as follows: in \cref{sec:neg-cycle:sec:blackbox-reduction} we give the reduction from
finding negative cycles to \textsc{Threshold}. In~\cref{sec:neg-cycle:sec:slow-threshold} we give an implementation of
\textsc{Threshold} which has an extra log-factor compared to the promised \cref{thm:sssp-or-negative-cycle}, but it has the benefit
of being simple. Finally, in~\cref{sec:neg-cycle:sec:fast-threshold} we give a faster (but more involved) implementation of \textsc{Threshold} which
yields~\cref{thm:sssp-or-negative-cycle}.

\subsection{Reduction to Threshold} \label{sec:neg-cycle:sec:blackbox-reduction}

In this section we restate the reduction given by Bernstein et al. in~\cite[Section 7.1]{BernsteinNW22} from finding
a negative cycle if it exists, to \textsc{Threshold} and \RestrictedSSSP{} (see their algorithm \textsc{SPLasVegas}).

\begin{lemma}[Finding Negative Cycles] \label{lem:reduction-to-threshold}
Let $G$ be a directed graph with a negative cycle. There is a Las Vegas algorithm $\textsc{FindNegCycle}(G)$ which finds a negative cycle in $G$, and runs in time $\Order(T_{\TRestrictedSSSP}(m,n) \log(nW) + T_{\textsc{Threshold}}(m,n))$ with high probability.
\end{lemma}

\begin{algorithm}[t]
\caption{The procedure to find a negative cycle. Given a graph $G$ containing a negative cycle, it finds one such negative cycle $C$. See~\cref{lem:reduction-to-threshold}.} \label{alg:find-negative-cycle}
\smallskip
\begin{algorithmic}[1]
\Procedure{FindNegCycle}{$G$}
    \State Let $G_0$ be a copy of $G$ with $w_{G_0}(e) = (n^3 + 1) \cdot w_G(e)$ \label{alg:find-negative-cycle:line:scaling}
    \State Let $M^* \gets \textsc{Threshold}(G_0)$ \label{alg:find-negative-cycle:line:threshold}
    \State Let $G_1$ be a copy of $G_0^{+M^*}$ where we add an artificial source
    \Statex[2] vertex $s$ to $G_1$ with weight-$0$ edges to all other vertices
    \State Run $\textsc{SSSP}(G_1, s)$ (see~\cref{alg:sssp}) and set $\phi(v) = \dist_{G_1}(s, v)$ \label{alg:find-negative-cycle:line:sssp}
    \State Let $G_2$ be the graph $(G_1)_\phi$ where we remove all edges with weight greater than $n$
    \If{there is a cycle $C$ in $G_2$} \label{alg:find-negative-cycle:line:cycle}
        \If{$C$ is negative in $G$}{ \Return $C$}\EndIf
    \EndIf
    \State \Return $\textsc{FindNegCycle}(G)$ \emph{(i.e., restart)}
\EndProcedure
\end{algorithmic}
\end{algorithm}

\begin{proof}
See the pseudocode in~\cref{alg:find-negative-cycle} for a concise description. We start by
defining a graph $G_0$ which is a copy of $G$ but with edge weights multiplied by $n^3 + 1$.
Then we compute~$M^*$ using $\textsc{Threshold}(G_0)$, and let $G_1$ be $G_0^{+M^*}$.
Next, we add a dummy source $s$ to $G_1$ connected with 0-weight edges to all other vertices, and
run $\textsc{SSSP}$ on the resulting graph from $s$. We then use the distances computed to
construct a potential $\phi$, and construct a graph $G_2$ by applying the potential $\phi$ to $G_1$ and subsequently
removing all the edges with weight larger than $n$. Finally, we check if $G_2$ contains any cycle (of any weight) and if so,
check it has negative weight in the original graph $G$ and return it. Otherwise, we restart the algorithm from the beginning.

The correctness is obvious: When the algorithm terminates, it clearly returns a negative cycle. The interesting part is to show that with high probability the algorithm finds a negative cycle $C$ without restarting. The call to $\textsc{Threshold}(G_0)$ in~\cref{alg:find-negative-cycle:line:threshold} returns the smallest $M^* \geq 0$ such that $G_0$ contains no negative cycle, with high probability. In this case, by definition, $G_1$ does not contain a negative cycle, and therefore by~\cref{thm:sssp} the call to $\textsc{SSSP}(G_1, s)$ correctly computes a shortest path tree from $s$. From now on, we condition on these two events.

\begin{claim}\label{proof:reduction-to-threshold:claim:1}
    It holds that $M^* > n^2$.
\end{claim}
\begin{claimproof}
    Let $C$ be a simple cycle in $G$ with minimum (negative) weight. Since $G_1 = G_0^{+M^*}$ contains no negative cycles, the weight of $C$ in $G_1$ is $0 \leq w_1(C) = w_0(C) + M^*|C|$. The claim follows by noting that $w_0(C) < -n^3$ due to the scaling in~\cref{alg:find-negative-cycle:line:scaling}, and that $|C| \leq  n$ because $C$ is simple.
\end{claimproof}

 Next, we argue that a cycle of minimum weight in $G$ remains a cycle in $G_2$, and conversely that any simple cycle in $G_2$ corresponds to a negative weight cycle in $G$. Note that this is enough to prove that the algorithm terminates with high probability without a restart.

\begin{claim}
    Let $C$ be a simple cycle in $G$ of minimum weight. Then, $C$ is a cycle in $G_2$.
\end{claim}
\begin{claimproof}
    First note that the weight of $C$ in $G_0^{+M^*}$ (and thus also in $G_1$) is at most $n$. This holds since $M^*$ is the smallest integer such that $G_0^{+M^*}$ contains no negative cycles, which means that $w_0(C) - |C| < 0$. Second, note that since \cref{alg:find-negative-cycle:line:sssp} correctly computes a shortest path tree in $G_1$, it holds that the edge weights in $(G_1)_\phi$ are all non-negative (by~\cref{cor:johnson-non}). Moreover, the weight of $C$ in $(G_1)_\phi$ is the same as in $G_1$ (by~\cref{lem:johnson}). Thus, we conclude that the removal of the edges of weight greater than $n$ in $(G_1)_{\phi}$ to obtain $G_2$ leaves $C$ untouched.
\end{claimproof}

\begin{claim}
    Any cycle $C'$ in $G_2$ has negative weight in $G$.
\end{claim}
\begin{claimproof}
    Note that $w_2(C') \leq n^2$ since every edge in $G_2$ has weight at most $n$. Moreover, since~$G_2$ is obtained from $G_1$ by adding a potential, it holds that $w_2(C') = w_1(C')$ (by~\cref{lem:johnson}). Therefore, $w_0(C') = w_1(C') - M^*|C'| \leq n^2 - M^* < 0$ where the last inequality holds since $M^* > n^2$ by~\cref{proof:reduction-to-threshold:claim:1}.
\end{claimproof}

Finally, we analyze the running time. The call to $\textsc{Threshold}(G_0)$ succeeds with high probability (see~\cref{def:threshold}).
Conditioned on this, $G_1$ contains no negative cycles. Thus by~\cref{thm:sssp}, the call to $\textsc{SSSP}(G, s)$ runs
in time $\Order(T_{\TRestrictedSSSP}(m,n) \log(nW))$ with high probability. Note that the remaining steps of the algorithm take
time $\Order(m)$. Therefore, we conclude that the overall running time is $\Order(T_{\TRestrictedSSSP}(m,n) \log(nW) + T_{\textsc{Threshold}}(m,n))$
with high probability.
\end{proof}

\subsection{Simple Implementation of Threshold} \label{sec:neg-cycle:sec:slow-threshold}

In this section we give a simple implementation of \textsc{Threshold} which combined with~\cref{lem:reduction-to-threshold} yields an algorithm to find negative cycles in time $\Order(T_{\TRestrictedSSSP}(m,n) \log n \log(n W))$. This procedure shaves one log-factor compared to~\cite{BernsteinNW22} (see their algorithm \textsc{FindThresh} in Lemma~7.2). Later, in~\cref{sec:neg-cycle:sec:fast-threshold}, we give an improved but more intricate algorithm.

As a building block, we will use the routine \textsc{Scale} from~\cref{lem:scale}.
The following lemma boosts the probability of success of \textsc{Scale} and uses a different parameterization of the minimum weight in the input graph, which will streamline our presentation.

\begin{lemma}[Test Scale]\label{lem:scale-test}
    Let $G$ be a directed graph with minimum weight at least $-W$ where $W \geq 24$, and  let $0 < \delta < 1$ be a parameter.
    There is an algorithm $\textsc{TestScale}(G, \delta)$ with the following properties:
    \begin{itemize}
        \item If $G$ does not contain a negative cycle, then with probability at least $1 - \delta$ it succeeds and returns a potential $\phi$ such that $G_{\phi}$ has minimum weight at least $-\tfrac{3}{4}W$. If it does not suceed, it returns \textsc{Fail}.
        \item It runs in time $\Order(T_{\TRestrictedSSSP}(m,n) \cdot \log(1/\delta))$.
    \end{itemize}
\end{lemma}
\begin{proof}
We run $\textsc{Scale}(G)$ (see~\cref{lem:scale}) for $\Order(\log(1/\delta))$ repetitions. Each execution either returns a potential
$\phi$, or it fails. We return \textsc{Fail} if and only if \emph{all} these repetitions fail.
The running time analysis is immediate by~\cref{lem:scale}.

Now we analyze correctness. First we look at the success probability.
\cref{lem:scale} guarantees that if $G$ does not contain a negative cycle, then each invocation to $\textsc{Scale}(G)$ returns a potential $\phi$ with constant probability.
Thus, in this case, the probability that all $\Order(\log(1/\delta))$ repetitions fail and we return \textsc{Fail} is at most $\delta$, as stated.
Next, we analyze the increase in the minimum weight of $G_{\phi}$.
Recall that the minimum weight in $G$ is at least $-W$.
Let $k$ be the largest integer such that $W \geq 3k$, and let $-W'$ denote the minimum weight of $G_\phi$.
In particular, the minimum weight in $G$ is greater than $-3(k+1)$, so \cref{lem:scale} guarantees that
\begin{equation*}
    -W' > -2(k+1) \geq -\tfrac{2}{3}W - 2 \geq -\tfrac{2}{3}W - \tfrac{1}{12}W = -\tfrac{3}{4}W,
\end{equation*}
where the last inequality uses the assumption that $W \geq 24$.
\end{proof}

\begin{algorithm}[t]
\caption{The slow implementation of \textsc{Threshold}. Given a graph $G$, it computes the smallest
integer $M^* \geq 0$ such that $G^{+M^*}$ contains no negative cycle. See~\cref{lem:slow-threshold}.} \label{alg:slow-threshold}
\smallskip
\begin{algorithmic}[1]
\Procedure{SlowThreshold}{$G$}
    \State Let $-W$ be the smallest weight in $G$
    \If{$W \leq 48$}
        \For{$t \gets 47, \dots, 1, 0$}
            \If{$\textsc{DetectNegCycle}(G^{+t}) = \textsc{NegCycle}$} \Return $t + 1$
            \EndIf
        \EndFor
        \State \Return $0$
    \Else
        \State Let $M \gets \ceil{\frac W2}$
        \If{$\textsc{TestScale}(G^{+M}, n^{-10}) = \phi$} \Return $\textsc{SlowThreshold}(G_\phi)$
        \Else\ \Return $M + \textsc{SlowThreshold}(G^{+M})$
        \EndIf
    \EndIf
\EndProcedure
\end{algorithmic}
\end{algorithm}

\begin{lemma}[Slow Threshold]\label{lem:slow-threshold}
Let $G$ be a directed graph. There is an algorithm computing $\textsc{Threshold}(G)$ (\cref{def:threshold}) which succeeds with high probability and runs in worst-case time $\Order(T_{\TRestrictedSSSP}(m,n) \log n \log(n W))$.
\end{lemma}
\begin{proof}
We summarize the pseudocode in \cref{alg:slow-threshold}. Let $-W$ be the smallest weight in~$G$. If $W \leq 48$ (i.e., all weights are at least $-48$) we clearly have that the correct answer lies in the range $0 \leq M^* \leq 48$. We brute-force the answer by exhaustively checking which graph $G^{+47}, \dots, G^{+0}$ is the first one containing a negative cycle. For this test we use the algorithm $\textsc{DetectNegCycle}(G)$. \cref{cor:detect-negcycle} guarantees that it reports correct answers with high probability.

If $W > 48$, we make progress by reducing the problem to another instance with larger minimum weight. Let $M = \ceil{\frac W2}$, and run $\textsc{TestScale}(G^{+M}, \delta)$ for $\delta := 1/n^{10}$. We distinguish two cases based on the outcome of \textsc{TestScale}:

\begin{itemize}
    \item Case 1: $\textsc{TestScale}(G^{+M}, \delta) = \phi$ for a potential function $\phi$. Then recursively compute and return $\textsc{SlowThreshold}(G_\phi)$. First note that this is correct, i.e., that the answer is unchanged by recursing on $G_\phi$, since the potential does not change the weight of any cycle (see \cref{lem:johnson}). Second, note that we make progress by increasing the smallest weight in $G_\phi$ to least $-\frac{11}{12} W$:
    To see this, note that the minimum weight of $G^{+M}$ is at least $-\tfrac{1}{2}W$, and thus, \cref{lem:scale-test} guarantees that the smallest weight in $G_\phi^{+M}$ is at least $-\frac{3}{8} W$.
    Therefore, it follows that the smallest weight in $G_{\phi}$ is at least
    \[
        -\tfrac{3}{8}W - M = \tfrac{3}{8}W - \ceil{\tfrac{1}{2}W}
            \geq -\tfrac{7}{8}W - 1
            > -\tfrac{7}{8}W - \tfrac{1}{24}W = -\tfrac{11}{12}W,
    \]
    where the second inequality uses the assumption that $W > 24$.
    \item Case 2: $\textsc{TestScale}(G^{+M}, \delta) = \textsc{Fail}$. By \cref{lem:scale-test}, if $G^{+M}$ does not contain a negative cycle then with high probability the output is not \textsc{Fail}. Conditioned on this event, we conclude that $G^{+M}$ contains a negative cycle. Thus, we know that the optimal answer $M^*$ satisfies $M^* \geq M$, and therefore we return $M + \textsc{SlowThreshold}(G^{+M})$.
    Note that this also improves the most negative edge weight to $-W + M \geq -\tfrac{11}{12}W$.
\end{itemize}

We claim that the running time is bounded by $\Order(T_{\TRestrictedSSSP}(m,n) \log n \log(nW))$. To see this, note that in the base case, when $W \leq 48$, the algorithm calls $\textsc{DetectNegCycle}(G)$ and therefore takes time $\Order(T_{\TRestrictedSSSP}(m,n) \cdot \log(nW))$ (see~\cref{cor:detect-negcycle}). We claim that the higher levels of the recursion take time $\Order(T_{\TRestrictedSSSP}(m,n) \log n \log W)$ in total. Note that each such level takes time $\Order(T_{\TRestrictedSSSP}(m,n) \cdot \log n)$ due to the call to \textsc{TestScale} (\cref{lem:scale-test}) and thus, it suffices to bound the recursion depth by $\Order(\log W)$. To this end, observe that we always recur on graphs for which $W$ has decreased by a constant factor.

Finally note that each call to \textsc{TestScale} succeeds with high probability, and we make one call for each of the $\Order(\log W)$ recursive calls. Thus, by a union bound the algorithm succeeds with high probability. (Strictly speaking, for this union bound we assume that~$\log W \leq n$; if instead $\log W > n$, we can simply use Bellman-Ford's algorithm.)
\end{proof}

\subsection{Fast Implementation of Threshold}\label{sec:neg-cycle:sec:fast-threshold}

In this section we give the fast implementation of \textsc{Threshold}.

\begin{lemma}[Fast Threshold]\label{lem:threshold}
Let $G$ be a directed graph. There is an algorithm computing $\textsc{Threshold}(G)$ (see~\cref{def:threshold})
which suceeds with high probability, and runs in worst-case time $\Order(T_{\TRestrictedSSSP}(m,n) \log(n W))$.
\end{lemma}

The algorithm is intricate, so we start with a high level description to convey some intuition.

\subparagraph*{High-Level Idea}

Let $\Delta$ be a parameter and let $M^* \geq 0$ be the right threshold. Let us look at what happens if we make a call to $\textsc{TestScale}(G^{+W - \Delta}, \delta)$, where $1 - \delta$ is the success probability and $-W$ is the minimum edge weight in $G$. If $G^{+W - \Delta}$ does not have negative cycles, then \cref{lem:scale-test} guarantees that with probability at least $1-\delta$ we obtain a potential~$\phi$. On the other hand, if $G^{+W - \Delta}$ contains a negative cycle, then we have \emph{no guarantee} from \cref{lem:scale-test}. That is, the algorithm might return a potential, or it might return \textsc{Fail}. The upside is that as long as we obtain a potential, regardless whether there is a negative cycle or not, we can make progress by (additively) increasing the minimum edge weight by $\approx \Delta$. Moreover, if we obtain \textsc{Fail}, then we conclude that with probability at least $1 - \delta$ the graph~$G^{+W - \Delta}$ contains a negative cycle. This suggests the following idea. We make a call to $\textsc{TestScale}(G^{+W - \Delta}, \delta)$, and consider the two outcomes:
\medskip
\begin{enumerate}
    \item $\textsc{TestScale}(G^{+W - \Delta}, \delta) = \phi$. Then, we set $G := G_\phi$ and increase $\Delta := 2\Delta$.
    \item $\textsc{TestScale}(G^{+W - \Delta}, \delta) = \textsc{Fail}$. Then, we decrease $\Delta := \Delta/2$.
\end{enumerate}
\medskip
If we are in Case 1, then the minimum edge weight $-W'$ of $G_\phi$ is increased by $\Delta$. This in turn, decreases the gap $W' - M^*$ (note that at all times $M^* \leq W'$). Thus, larger $\Delta$ implies larger progress in decreasing $W' - M^*$. This is why in this case we double $\Delta$. On the other hand, if we are in Case 2 then by the guarantee of~\cref{lem:scale-test}, we conclude that with probability at least $1 - \delta$ the graph $G^{+W - \Delta}$ contains a negative cycle. Intuitively, this means that $\Delta$ is too large. Therefore, we halve $\Delta$ to eventually make progress in Case~1 again.

In short, we know that when $G^{+W - \Delta}$ does not have negative cycles, or equivalently $W - M^* \geq \Delta$, then with probability at least $1-\delta$ we will make progress in Case~1 by decreasing the gap $W - M^*$. On the other hand, if we are in Case 2 and $G^{+W - \Delta}$ has a negative cycle, or equivalently $W - M^* < \Delta$, then we will make progress by decreasing $\Delta$.

Perhaps surprisingly, we will show that this idea can be implemented by choosing $\delta = 0.01$, and not $1/\poly(n)$ as in the implementation of~\cref{lem:slow-threshold} (which was the reason for getting an extra $\Order(\log n)$-factor there). For this, we will formalize the progress as some \emph{drift function} that decreases in expectation in each iteration, and then apply a \emph{drift theorem} (see~\cref{thm:drift-tailbound}).

\subparagraph*{The Algorithm}
Now we formalize this approach. We proceed in an iterative way. At iteration $t$, we have a graph $G_t$ with minimum weight $-W_t$, and we maintain a parameter $\Delta_t$. We make a call to $\textsc{Scale}(G^{+W_t - \Delta_t}, \delta)$ with $\delta := 0.01$. If we obtain a potential $\phi$ as answer, we set $G_{t+1} := (G_t)_{\phi}$ and $\Delta_{t+1} := 2\Delta$. Otherwise, we set $G_{t+1} := G_t$ and $\Delta_{t+1} := \tfrac12\Delta_t$. After $T = \Theta(\log(nW))$ iterations, we stop and return $W_T$ as the answer. The complete pseudocode (which additionally handles some corner cases) is in~\cref{alg:threshold}.

\begin{algorithm}[t]
\caption{The fast implementation of \textsc{Threshold}. Given a graph $G$, it computes the smallest
integer $M^* \geq 0$ such that $G^{+M^*}$ contains no negative cycle. See~\cref{lem:threshold}.}\label{alg:threshold}
\smallskip
\begin{algorithmic}[1]
\Procedure{FastThreshold}{$G$}
\State Let $G_0 \gets G$ and $\Delta_0 \gets 2$
\State $T \gets \Theta(\log(nW))$ with sufficiently large hidden constant
\For{$t \gets 0,\dots,T-1$} \label{alg:threshold:line:for-loop}
    \State Let  $-W_t$ be the smallest edge weight in $G_t$
    \If{$W_t \leq 24$}
        \For{$j \gets 23, \dots, 1, 0$}\label{alg:threshold:line:for-solve-directly}
            \If{$\textsc{FindNegCycle}(G_t^{+j}) = \textsc{NegCycle}$} \Return $j + 1$ \label{alg:threshold:line:solve-directly}
            \EndIf
        \EndFor
    \EndIf
    \If{$\textsc{TestScale}(G_{t}^{+W_{t} - \Delta_{t}}, 0.01) = \phi$} \label{alg:threshold:line:test-scale}
        \State $G_{t+1} \gets (G_t)_{\phi}, \; \Delta_{t+1} \gets 2 \cdot \Delta_t$
    \Else
        \State $G_{t+1} \gets G_t, \; \Delta_{t+1} \gets \max\{1, \tfrac{\Delta_t}{2}\}$\label{alg:threshold:line:no-potential}
    \EndIf
\EndFor
\State \Return $W_T$
\EndProcedure
\end{algorithmic}
\end{algorithm}

To quantify the progress made by the algorithm, we define the following
\emph{drift function} at iteration $t$:
\begin{equation}
    D_t := (W_t - M^*)^{20} \cdot \max\left \{\frac{2\Delta_t}{W_t - M^*}, \frac{W_t - M^*}{2\Delta_t}\right \}, \label{eqn:drift}
\end{equation}
Observe that we always have $\Delta_t \ge 1$ and $W_t \ge M^*$ throughout the algorithm. To cover the case $W_t = M^*$ (where the above expression leads to a division by 0), formally we actually define the drift function by
\begin{equation}
    D_t := \max\left \{(W_t - M^*)^{19} \cdot 2\Delta_t, \frac{(W_t - M^*)^{21}}{2\Delta_t}\right \}. \label{eqn:drifttwo}
\end{equation}
For the sake of readability, in the following we work with (\ref{eqn:drift}), with the understanding that formally we mean (\ref{eqn:drifttwo}).

We will show that $D_t$ decreases by a constant factor (in expectation) in each iteration of the for-loop
in~\cref{alg:threshold:line:for-loop}. Note that when $D_t$ reaches 0, then
we have that $W_t = M^*$, so we are done.

\begin{lemma}[Negative Drift]\label{lem:negative-drift}
    For any $d > 0$ and $t \geq 0$ it holds that
    \[
        \Ex(D_{t+1}\mid D_t=d) \leq 0.7\cdot d.
    \]
\end{lemma}

Before proving~\cref{lem:negative-drift}, let us see how to obtain \cref{lem:threshold} from it. For this, we will use the
following tool:

\begin{theorem}[{Multiplicative Drift, see e.g.~\cite[Theorem 18]{Lengler17}}]\label{thm:drift-tailbound}
    Let $(X_t)_{t \geq 0}$ be a sequence of non-negative random variables with a finite state space $\mathcal{S}$
    of non-negative integers.
    Suppose that $X_0 = s_0$, and there exists $\delta > 0$ such that for all
    $s \in \mathcal{S} \setminus \{0\}$ and all $t \geq 0$, $\Ex(X_{t+1} \mid X_t=s) \leq (1-\delta) s$. Then, for all $r \geq 0$,
    \[
        \Pr(X_r > 0) \leq e^{-\delta \cdot r} \cdot s_0.
    \]
\end{theorem}
\begin{proof}
    By Markov's inequality, $\Pr(X_r > 0) = \Pr(X_r \geq 1) \leq \Ex(X_r)$. By applying the bound
    $\Ex(X_{t+1} \mid X_t=d) \leq (1-\delta) d$ for $r$ times, we obtain that
    \[
        \Pr(X_r > 0) \leq (1 - \delta)^r \cdot s_0 \leq \exp(-\delta r) \cdot s_0. \qedhere
    \]
\end{proof}

\begin{proof}[Proof of~\cref{lem:threshold}]
    See~\cref{alg:threshold} for the pseudocode. First we analyze the running time. During each iteration of the
    for-loop, it either holds that $W_t \leq 24$ and we solve the problem directly using at most 24 calls to \textsc{DetectNegCycle},
    or we make a call to \textsc{TestScale}. Each call to \textsc{TestScale} takes time
    $\Order(T_{\TRestrictedSSSP}(m,n))$ by~\cref{lem:scale-test}, and we only make the calls to \textsc{DetectNegCycle} once which
    take total time $\Order(T_{\TRestrictedSSSP}(m,n)\log n)$ by~\cref{cor:detect-negcycle}. Since $T = \Theta(\log(nW))$, the overall
    running time is bounded by $\Order(T_{\TRestrictedSSSP}(m,n)\log n + T_{\TRestrictedSSSP}(m,n) \log(nW))$, as claimed.

    Now we analyze correctness. Note that at every iteration, $G_t$ is equivalent to $G$ since the only way we modify
    the graph is by adding potentials (see~\cref{lem:johnson}). Thus, if at some point we have that $W_t \leq 24$ then
    the correct answer lies in the range $0 \leq M^* \leq 24$. The for-loop in~\cref{alg:threshold:line:for-solve-directly}
    exhaustively checks which is the correct value by making calls to \textsc{DetectNegCycle}. By~\cref{cor:detect-negcycle}, this is correct
    with high probability.

    Now suppose the algorithm does not terminate in~\cref{alg:threshold:line:solve-directly}.
    We claim that the final drift $D_T$ is zero with high probability. Note that this implies correctness,
    since $D_T = 0$ if and only if $W_T = M^*$ (to see this, observe that $\Delta_T \geq 1$ due to~\cref{alg:threshold:line:no-potential}).
    To prove the claim, we will use~\cref{thm:drift-tailbound}. Note that \cref{lem:negative-drift} gives us that $\Ex(D_{t+1}\mid D_t=d) \leq 0.7 d$.
    Moreover, we can bound the initial drift $D_0$ as
    \[
        D_0 = (W - M^*)^{20} \cdot \max \left \{\frac{2\Delta_0}{W - M^*}, \frac{W - M^*}{2\Delta_0}\right \}
            \leq (W-M^*)^{21}\cdot 2\Delta_0 \leq 4W^{21}.
    \]
    Hence, \cref{thm:drift-tailbound} yields that $\Pr(D_T > 0) \leq \exp(-0.7 T) \cdot 4W^{21}$.
    Since $T = \Theta(\log(nW))$, we conclude that $\Pr(D_T > 0) \leq n^{-\Omega(1)}$, which finishes the proof.
\end{proof}

\begin{proof}[Proof of~\cref{lem:negative-drift}]
    Focus on iteration $t$ of the for-loop in~\cref{alg:threshold:line:for-loop}.
    Let $E_1$ be the event that we obtain a potential $\phi$ (i.e. that the if-statement in \cref{alg:threshold:line:test-scale}
    succeeds) and let $E_2 := \neg E_1$ be the complement. We start by observing how the parameters $W_{t+1}$ and $\Delta_{t+1}$
    change depending on whether $E_1$ or $E_2$ occur.
    \begin{claim}\label{proof:negative-drift:claim:e1}
        If $E_1$ occurs, then $W_{t+1} \leq W_t - \tfrac{\Delta_t}{4}$, and $\Delta_{t+1} = 2\Delta_t$.
    \end{claim}
    \begin{claimproof}
        If the call to $\textsc{TestScale}$ in \cref{alg:threshold:line:test-scale} returns a potential $\phi$,
        then we set $G_{t+1} = (G_t)_{\phi}$ and $\Delta_{t+1} = 2 \Delta_t$.
        Observe that the minimum weight of $G_t^{+W_t-\Delta_t}$ is $\Delta_t$.
        Hence, \cref{lem:scale-test} guarantees that the minimum weight of $(G_t)_{\phi}^{+W_t-\Delta_t}$ is at least $-\tfrac{3}{4}\Delta_t$.
        Since $G_{t+1} = (G_t)_{\phi}$ is defined by substracting $W_t - \Delta_t$ from every edge weight in
        $(G_t)_{\phi}^{+W_t - \Delta_t}$, we obtain that $-W_{t+1} \geq -W_t + \tfrac{1}{4}\Delta_t$.
    \end{claimproof}

    \begin{claim}\label{proof:negative-drift:claim:e2}
        If $E_2$ occurs, then $W_{t+1} = W_t$ and $\Delta_{t+1} = \max\{1, \Delta_t/2\}$ and $D_{t+1} \le 2 D_t$.
    \end{claim}
    \begin{claimproof}
        The first two statements are immediate by~\cref{alg:threshold:line:no-potential}. Towards the third statement, for the function $f(x) := \max\{x,1/x\}$ we observe that if $x,y>0$ differ by at most a factor~2 then also $f(x),f(y)$ differ by at most a factor 2. Now we use that $D_t = (W_t - M^*)^{20} \cdot f( 2\Delta_t / (W_t-M^*))$. Since $\Delta_t \ge 1$, it holds that $\Delta_t,\Delta_{t+1}$ differ by at most a factor 2, and thus $D_t,D_{t+1}$ differ by at most a factor 2.
    \end{claimproof}

    With these claims, we proceed to bound the drift $D_{t+1}$ when $D_t > 0$. Recall that we defined
    \begin{equation*} \tag{\ref{eqn:drift}}
      D_t = (W_t - M^*)^{20} \cdot \max\left \{\frac{2\Delta_t}{W_t - M^*}, \frac{W_t - M^*}{2\Delta_t}\right \}.
    \end{equation*}
    Note that it always holds that $W_t \geq M^*$ and $W_{t+1} \geq M^*$. Moreover, since $D_t > 0$, we can assume that
    $W_t - M^* > 0$, since otherwise $W_t - M^* = 0$ and hence $D_t = 0$.
    We proceed making a case distinction based on the term that achieves the maximum in~\eqref{eqn:drift}.

    \begin{description}
        \item[Case 1] $\Delta_t \geq \tfrac{1}{2}(W_t - M^*)$: Then, we have that $D_t = (W_t - M^*)^{19} \cdot 2\Delta_t$.
        If $E_1$ occurs, then by~\cref{proof:negative-drift:claim:e1} it holds that $\Delta_{t+1} \geq \Delta_t \geq \tfrac{1}{2}(W_{t} - M^*) \geq \tfrac{1}{2}(W_{t+1} - M^*)$.
        Therefore, using~\eqref{eqn:drift} we can bound the drift $D_{t+1}$ by
        \begin{align*}
            D_{t+1} &= (W_{t+1} - M^*)^{19} \cdot 2 \Delta_{t+1} \\
                &\leq (W_t - M^* - \tfrac{\Delta_t}{4})^{19} \cdot 4\Delta_t \\
                &\leq (W_t - M^* - \tfrac{1}{8}(W_t - M^*))^{19} \cdot 4\Delta_t \\
                &\leq (\tfrac{7}{8})^{19} \cdot 2 D_t \leq 0.16 D_t,
        \end{align*}
        where we used~\cref{proof:negative-drift:claim:e1} in the first inequality, and the
        second inequality follows since by the assumption of Case 1 we have that $\tfrac{\Delta_t}{4} \geq \tfrac{1}{8}(W_t - M^*)$.
        \smallskip

        If $E_2$ occurs instead, we make a further case distinction:
        \begin{description}
            \item[Case 1.1] $\Delta_t > W_t - M^*$:
            Note that if $\Delta_t = 1$, then since $W_t$ and $M^*$ are integers
            it follows that $W_t = M^*$, and consequently $D_t = 0$, which contradicts the assumption that $D_t > 0$.
            Therefore, we can assume that $\Delta_t \geq 2$. In particular, by~\cref{proof:negative-drift:claim:e2} we have
            $\Delta_{t+1} = \frac12 \Delta_t > \frac 12 (W_t - M^*) = \frac 12 (W_{t+1} - M^*)$. Thus, by~\eqref{eqn:drift} we can express the drift $D_{t+1}$ as
            \[
                D_{t+1} = (W_{t+1} - M^*)^{19} \cdot 2\Delta_{t+1} = (W_t - M^*)^{19} \cdot \Delta_t = \frac{D_t}{2}.
            \]
            \item[Case 1.2] $\Delta_t \leq W_t - M^*$:
            Observe that in this case $G^{+W_t - \Delta_t}$ contains no negative cycle. Moreover, we can assume that $W_t > 24$ since otherwise
            the problem is solved directly in~\cref{alg:threshold:line:for-solve-directly}.
            Therefore, by \cref{lem:scale-test} we have that $\Pr(E_2) \leq 0.01$.
            Moreover, by by~\cref{proof:negative-drift:claim:e2} we have $D_{t+1} \leq 2D_t$.
        \end{description}

        \smallskip
        Combining the above, we conclude that for Case 1 it holds that
        \begin{align*}
                \Ex(D_{t+1} \mid D_t) &\leq \Pr(E_1) \Ex(D_{t+1}\mid D_t, E_1) + \Pr(E_2)\Ex(D_{t+1}\mid D_t, E_2) \\
                &\leq 1\cdot 0.16D_t + \max\left\{1\cdot \tfrac12 D_t, 0.01 \cdot 2 D_t\right\} \le 0.66 D_t.
        \end{align*}
        \smallskip

        \item[Case 2] $\Delta_t < \tfrac{1}{2}(W_t - M^*)$: Then, it holds that $D_t = (W_t - M^*)^{21}/(2\Delta_t)$.
        If $E_2$ occurs, then by the same argument as in Case 1.2 we have that $D_{t+1} \leq 2D_t$ and $\Pr(E_2) \leq 0.01$.

        If $E_1$ occurs instead, then we make a further case distinction:
        \begin{description}
            \item[Case 2.1] $\Delta_{t+1} < \tfrac{1}{2}(W_{t+1} - M^*)$: Then using~\eqref{eqn:drift}, it holds that
            \[
                D_{t+1} = \frac{(W_{t+1}-M^*)^{21}}{2\Delta_{t+1}} \leq \frac{(W_t - M^*)^{21}}{4\Delta_t} = \frac{D_t}{2},
            \]
            where the inequality holds due to~\cref{proof:negative-drift:claim:e1}.

            \item[Case 2.2] $\Delta_{t+1} \geq \tfrac{1}{2}(W_{t+1} - M^*)$: Then it holds that
            $D_{t+1} = (W_{t+1} - M^*)^{19} \cdot 2\Delta_{t+1}$. Since by the assumption of Case 2 we have
            $(W_t - M^*)/(2\Delta_t) \geq 1$ and by~\cref{proof:negative-drift:claim:e1} we have $\Delta_{t+1} = 2\Delta_t$, we can bound $D_{t+1}$ as
            \begin{align}
                D_{t+1} &= (W_{t+1} - M^*)^{19} \cdot 2\Delta_{t+1} \nonumber \\
                    &\leq (W_{t+1} - M^*)^{19} \cdot 4\Delta_t \cdot \left(\frac{W_t - M^*}{2\Delta_t}\right)^2 \nonumber \\
                    &= (W_{t+1} - M^*)^{19} \cdot (W_t - M^*)^2 \cdot \tfrac{1}{\Delta_t}. \label{proof:negative-drift:eqn:1}
            \end{align}
            By~\cref{proof:negative-drift:claim:e1}, we have that $W_{t+1} \leq W_t - \tfrac{\Delta_t}{4}$. Hence, can bound $W_{t+1}-M^*$ as
            \begin{align}
                W_{t+1} - M^*
                    &= \tfrac{16}{17}(W_{t+1}-M^*) + \tfrac{1}{17}(W_{t+1}-M^*) \nonumber \\
                    &\leq \tfrac{16}{17}(W_t - M^* - \tfrac{\Delta_t}{4}) + \tfrac{1}{17}(W_{t+1}-M^*) \nonumber \\
                    &= \tfrac{16}{17}(W_t - M^*) - \tfrac{16}{17} \cdot \tfrac{\Delta_t}{4} + \tfrac{1}{17}(W_{t+1}-M^*). \label{proof:negative-drift:eqn:2}
            \end{align}
            By \cref{proof:negative-drift:claim:e1} and the assumption of Case 2.2, we have that
            $2\Delta_t = \Delta_{t+1} \geq \tfrac{1}{2}(W_{t+1}-M^*)$. This implies that $\tfrac{\Delta_t}{4} \geq \tfrac{1}{16}(W_{t+1}-M^*)$.
            Plugging this into~\eqref{proof:negative-drift:eqn:2}, we obtain that
            \begin{align}
              W_{t+1} - M^*
                &\leq \tfrac{16}{17}(W_t - M^*) - \tfrac{1}{17}(W_{t+1} - M^*) + \tfrac{1}{17}(W_{t+1}-M^*) \nonumber\\
                &= \tfrac{16}{17}(W_t - M^*).\label{proof:negative-drift-eqn:3}
            \end{align}
            Finally, we combine \eqref{proof:negative-drift:eqn:1} and \eqref{proof:negative-drift-eqn:3} to obtain that
            \begin{align*}
                D_{t+1}
                    &\leq (W_{t+1} - M^*)^{19} (W_t - M^*)^2 \cdot \tfrac{1}{\Delta_t}\\
                    &\leq (\tfrac{16}{17})^{19} (W_t - M^*)^{21} \cdot \tfrac{1}{\Delta_t}\\
                    &= (\tfrac{16}{17})^{19} \cdot 2 \cdot D_t \\
                    &\leq 0.65 D_t
            \end{align*}
        \end{description}
        \smallskip
        Combining the subcases considered, we conclude that for Case 2 it holds that
        \begin{align*}
            \Ex(D_{t+1}\mid D_t) &\leq \Pr(E_1) \Ex(D_{t+1} \mid D_t, E_1) + \Pr(E_2)\Ex(D_{t+1}\mid D_t, E_2) \\
            &\leq 1 \cdot \max\left\{\tfrac 12 D_t, 0.65 D_t \right\} + 0.01 \cdot 2 D_t \leq 0.67 \cdot D_t.
        \end{align*}
    \end{description}
    Since cases 1 and 2 are exhaustive, the proof is concluded.
\end{proof}

\subsection{Putting Everything Together}

Now we put the pieces together to prove our main theorem.

\thmssspornegativecycle*
\begin{proof}
    The algorithm alternatingly runs the following two steps, and interupts each step after it exceeds a time budget of $\Order((m + n \log \log n) \log^2 n \log(nW))$:
    \begin{enumerate}
        \medskip
        \item Run $\textsc{SSSP}(G, s)$. If this algorithm finishes in time and returns a shortest path tree, we check that the shortest path tree is correct (by relaxing all edges and testing whether any distance in the tree changes) and return this shortest path tree in the positive case. Otherwise, we continue with step 2.
        \smallskip
        \item Run $\textsc{FindNegCycle}(G)$ (using~\cref{lem:threshold} to implement \textsc{Threshold}). If this algorithm finishes in time and returns a negative cycle, we verify that the output is indeed a negative cycle and return this negative cycle in the positive case. Otherwise, we continue with step 1.
    \end{enumerate}

    The algorithm is clearly correct: Whenever it terminates, it reports a correct solution. Let us focus on the running time. We distinguish two cases: First, assume that $G$ does \emph{not} contain a negative cycle. By \cref{thm:sssp} step 1 runs in time $\Order((m + n \log \log n) \log^2 n \log(nW))$ with high probability and is not interrupted in this case. Moreover, the SSSP algorithm returns a correct shortest path tree with high probability, and thereby terminates the algorithm after just one iteration of step 1.

    On the other hand, suppose that $G$ contains a negative cycle. The algorithm runs step 1 which is wasted effort in this case, but costs only time $\Order((m + n \log \log n) \log^2 n \log(nW))$. Afterwards, by \cref{lem:reduction-to-threshold,lem:threshold}, a single execution of step 2 runs within the time budget with high probability. Moreover, since \cref{lem:reduction-to-threshold} is a Las Vegas algorithm, it returns a true negative cycle and the algorithm terminates.

    The previous two paragraphs prove that the algorithm terminates after successively running step 1 and step 2 in time $\Order((m + n \log \log n) \log^2 n \log(nW))$ with high probability. Since we independently repeat these steps until the algorithm terminates, the same bound applies to the expected running time.
\end{proof}

Next, we prove \cref{thm:all-distances} using the previous \cref{thm:sssp-or-negative-cycle} as a black-box.

\thmalldistances*
\begin{proof}
First, remove all vertices from the graph not reachable from $s$ and return distance~$\infty$ for each such vertex. Then compute the set of strongly connected components $C_1,\ldots,C_\ell$ in~$G$ in time~$\Order(m + n)$. For every SCC $C_i$, run our SSSP algorithm from \cref{thm:sssp-or-negative-cycle} on $G[C_i]$ to detect whether it contains a negative cycle. For every vertex contained in a SCC with a negative cycle, we return distance $-\infty$ (as this SCC is reachable from $s$ and contains a negative cycle, we can loop indefinitely). Similarly, report $-\infty$ for all vertices reachable from one of the $-\infty$-distance vertices. After removing all vertices at distance $-\infty$, the remaining graph does no longer contain a negative cycle. We may therefore run the SSSP algorithm on the remaining graph to compute the missing distances.

Let $n_i$ and $m_i$ denote the number of vertices and edges in the subgraph $G[C_i]$. Then the total running time is
\begin{gather*}
    \Order\left(T_{\textsc{SSSP}}(m, n, W) + \sum_i T_{\textsc{SSSP}}(m_i, n_i, W)\right) \\
    \qquad= \Order\left(\left(m + n \log\log n + \sum_i m_i + \sum_i n_i \log\log n\right) \log n^2 \log(nW) \right) \\
    \qquad= \Order((m + n \log \log n) \log^2 n \log(nW)),
\end{gather*}
using that $\sum_i m_i \leq m$ and that $\sum_i n_i \leq n$.
\end{proof}

\section{Minimum Cycle Mean} \label{sec:mean-cycle}
In this section we prove~\cref{thm:mincyclemean}, i.e., we present the $\Order((m + n \log \log n)\log^2(n)\log(nW))$-time algorithm to compute the minimum cycle mean of a graph $G$.

Given a directed graph $G$, we denote by $\mu^*(G)$ the value of the minimum cycle mean, i.e.,~$\mu^*(G) := \min_C \bar w(C)$.
To develop our algorithm, the following characterization of the minimum cycle mean will be useful:

\begin{lemma}\label{lem:mincyclemean-characterization}
    Let $G$ be a directed graph. Then,
    \begin{equation*}
        \mu^*(G) = -\min\{Q \in \Rat \mid G^{+Q} \text{ contains no negative cycle}\}.
    \end{equation*}
\end{lemma}
\begin{proof}
By definition, we have that $\mu^*(G) = \min_C \bar w(C)$. Equivalently, $\mu^*(G)$ is the largest rational number $\mu$ such that $\mu \leq w(C)/|C|$ holds for all cycles $C$ in $G$. In particular, $w(C) - \mu \cdot |C| \geq 0$ holds for all cycles $C$, which is equivalent to $G^{-\mu}$ not having negative cycles.
\end{proof}

Recall that $\textsc{Threshold}(G)$, computes the minimum integer $M^* \geq 0$ such that $G^{+M^*}$ contains no negative cycle (\cref{def:threshold}). This is very similar to the characterization of the minimum cycle mean given by~\cref{lem:mincyclemean-characterization}, except that the latter minimizes over rational numbers that are not necessarily non-negative. To overcome this, we will use the following simple propositions:

\begin{proposition}\label{prop:shift-mean}
Let $G$ be a directed graph and let $a \geq 1, b \geq 0$ be integers. Let $H$ be a copy of $G$ where each edge has weight $w_H(e) := a \cdot w_G(e) + b$. Let $C$ be any cycle in $G$. Then, $\bar w_H(C) = a \cdot \bar w_G(G) + b$.
\end{proposition}
\begin{proof}
Note that the weight of $C$ in $H$ is exactly $w_H(C) = a \cdot w_G(C) + b \cdot |C|$. Therefore, the cycle mean of $C$ in $H$ equals $\bar w_H(C) = a \cdot w_G(C)/|C| + b = a \cdot \bar w_G(C) + b$.
\end{proof}

\begin{proposition}\label{prop:difference-means}
Let $C$ and $C'$ be two cycles in a directed graph $G$ with distinct means, i.e. $\bar w(C) \neq \bar w(C')$. Then, $|\bar w(C) - \bar w(C')| \geq 1/n^2$.
\end{proposition}
\begin{proof}
By definition, we can express $|\bar w(C) - \bar w(C')|$ as
\begin{equation*}
    \left|\frac{w(C)}{|C|}-\frac{w(C')}{|C'|} \right|
        = \left|\frac{w(C)|C'| - w(C')|C|}{|C|\cdot|C'|}\right| \geq \frac{1}{|C||C'|},
\end{equation*}
where we used that $\bar w(C) \neq \bar w(C')$. Since $|C|,|C'| \leq n$, we have that this is at least $1/n^2$.
\end{proof}

We will use the following lemma, which is a Las Vegas implementation of~\cref{lem:threshold}.

\begin{lemma}\label{lem:threshold-lasvegas}
Let $G$ be a directed graph. There is a Las Vegas algorithm which computes $\textsc{Threshold}(G)$ (see~\cref{def:threshold}) and runs in time $\Order((m + n \log\log n)\log^2n \log(nW))$ with high probability (and in expectation).
\end{lemma}
\begin{proof}
The algorithm computes $M^* = \textsc{Threshold}(G)$ using~\cref{lem:threshold}. By definition, this returns the smallest integer $M^*$ such that $G^{+M^*}$ contains no negative cycles with high probability (recall~\cref{def:threshold}). To turn it into a Las Vegas algorithm, we need to verify that the output is correct. For this, we add a source vertex $s$ connected with 0-weight edges to all other vertices and use~\cref{thm:sssp-or-negative-cycle} to test if $G^{+M^*}$ contains no negative cycles and that~$G^{+M^*-1}$ contains negative cycles. If either test fails, the algorithm restarts.

The correctness of this procedure follows since~\cref{thm:sssp-or-negative-cycle} is a Las Vegas algorithm. For the running time, observe that call to~\cref{lem:threshold} (using the bound on $T_{\TRestrictedSSSP}(m,n)$ of~\cref{thm:restricted-sssp}) and the calls to~\cref{thm:sssp-or-negative-cycle} run in time $\Order((m + n \log\log n)\log^2n \log(nW))$. Moreover, \cref{lem:threshold} guarantees that the value $M^*$ is correct with high probability. Thus, the algorithm terminates in $\Order((m + n \log\log n)\log^2n \log(nW))$-time with high probability.
\end{proof}

\begin{algorithm}[t]
\caption{Given a graph $G$ the procedure returns cycle $C$ of
minimum mean weight with high probability. See~\cref{thm:mincyclemean}.} \label{alg:min-cycle-mean}
\smallskip
\begin{algorithmic}[1]
\Procedure{MinCycleMean}{$G$}
    \State Let $L$ be the largest weight in $G$
    \State Let $H$ be a copy of $G$ with edge weights $w_H(e) \gets n^2 \cdot w_G(e) - n^3L$.
    \State Compute $M^* \gets \textsc{Threshold}(H)$ using~\cref{lem:threshold-lasvegas}\label{alg:min-cycle-mean:line:threshold}
    \State $C \gets \textsc{FindNegCycle}(H^{+M^* - 1})$\label{alg:min-cycle-mean:line:findnegcycle}
    \State \Return $C$
\EndProcedure
\end{algorithmic}
\end{algorithm}    

\thmmincyclemean*
\begin{proof}
We construct a graph $H$ by modifying each edge weight of $G$ to $n^2w(e) - n^3 L$, where $L$ is the largest edge-weight in $G$. Then, we compute $M^* := \textsc{Threshold}(H)$ using~\cref{lem:threshold-lasvegas}. Finally, we find a negative cycle in $H^{+M^* - 1}$ using~\cref{lem:reduction-to-threshold}. See~\cref{alg:min-cycle-mean} for the pseudocode.

The running time is dominated by the calls to \textsc{Threshold} and \textsc{FindNegCycle}. Using~\cref{lem:threshold-lasvegas} the call to \textsc{Threshold} takes time $\Order((m + n \log \log n) \log^2 n \log(nW))$ with high probability. By~\cref{lem:reduction-to-threshold} the call to \textsc{FindNegCycle}, (using~\cref{lem:threshold} to implement \textsc{Threshold} and \cref{thm:restricted-sssp} to bound $T_{\TRestrictedSSSP}(m,n)$) takes time $\Order((m + n \log \log n) \log^2 n \log(nW))$ with high probability as well. Thus, the algorithm runs in the claimed running time.

To analyze the correctness, note that~\cref{prop:shift-mean} implies that a cycle $C$ is the minimizer of $\bar w_G(C)$ if and only if it is the minimizer of $\bar w_H(C)$. Thus, it suffices to find a cycle of minimum mean in $H$. We will argue that the cycle found by the algorithm is the minimizer.

\begin{claim}\label{proof:min-cyclemean:claim:bound-meancycle}
The value $M^*$ computed in~\cref{alg:min-cycle-mean:line:threshold} satisfies
$M^* = \ceil{-\mu^*(H)}$.
\end{claim}
\begin{claimproof}
We observe that the minimum cycle mean in $H$ is non-positive, i.e., $\mu^*(H) \leq 0$. To see this, note that any cycle $C$ in $G$ has weight at most $w_G(C) \leq nL$. Thus, by the way we set the weights in $H$, any cycle in $H$ has weight $w_H(C) = n^2w_G(C) -n^3 L |C| \leq n^3L - n^3L = 0$. This means that in~\cref{lem:mincyclemean-characterization} we can minimize over $Q \geq 0$, i.e. that
\begin{equation}
    \mu^*(H) = -\min\{0 \leq Q \in \Rat \mid H^{+Q} \text{ contains no negative cycle}\}. \label{proof:min-cyclemean:eqn:1}
\end{equation}

Recall that by definition of $\textsc{Threshold}(H)$, $M^*$ is the smallest non-negative integer such that $H^{+M^*}$ has no negative cycles, i.e.
\begin{equation}
    M^* = \min\{0 \leq M \in \Int \mid H^{+M} \text{ contains no negative cycle}\}. \label{proof:min-cyclemean:eqn:2}
\end{equation}
Combining \eqref{proof:min-cyclemean:eqn:1} and \eqref{proof:min-cyclemean:eqn:2}, we conclude that $M^* = \ceil{-\mu^*(H)}$, as claimed.
\end{claimproof}

It follows that $H^{+M^*-1}$ indeed contains a negative cycle. By~\cref{lem:threshold-lasvegas}, the call to \textsc{Threshold} is correct. Hence, $H^{+M^*-1}$ contains a negative cycle and the call to \textsc{FindNegCycle} is correct by~\cref{lem:reduction-to-threshold}. Let $C$ be the cycle obtained in~\cref{alg:min-cycle-mean:line:findnegcycle}. Since it has negative weight in $H^{+M^*-1}$, its weight in $H$ is less than $-|C|(M^* - 1)$. Hence, it holds that $\bar w_H(C) < -M^* + 1$. Moreover, since $H^{+M^*}$ contains no negative cycle, every cycle $C'$ has mean weight $\bar w_H(C') \ge -M^*$.

Now consider a minimum mean cycle $C'$. As we have seen, we have
\begin{equation}\label{eq:cycle-mean-pm-one}
    -M^* \le \bar w_H(C') \le \bar w_H(C) < -M^* + 1.
\end{equation}
Assume for the sake of contradiction that $\bar w_H(C) \ne \bar w_H(C')$. Then by~\cref{prop:difference-means} we have that $|\bar w_G(C) - \bar w_G(C')| \geq 1/n^2$, and by~\cref{prop:shift-mean} it holds that $\bar w_H(C) = n^2 \cdot w_G(C) - n^3 L$ and $\bar w_H(C') = n^2 \cdot w_G(C') - n^3 L$. Combining these facts, we obtain that $|\bar w_H(C) - \bar w_H(C')| \geq 1$. This contradicts Equation~\eqref{eq:cycle-mean-pm-one}. Hence, we obtain $\bar w_H(C) = \bar w_H(C')$, so the computed cycle~$C$ is a minimizer of $\bar w_H(C)$ and thus also of $\bar w_G(C)$.
\end{proof}
\section{Low-Diameter Decompositions} \label{sec:ldd}
In this section we establish our strong Low-Diameter Decomposition (LDD). Recall that in a strong LDD (as defined in \cref{def:ldd}), the goal is to select a small set of edges $S$ such that after removing the edges in $S$, each strongly connected component in the remaining graph has bounded diameter. Our result is the following theorem, which proves that strong LDDs exist (which was known by~\cite{BernsteinNW22}) and can be efficiently computed (which was open):

\thmstrongldd*

\subsection{Heavy and Light Vertices}
In the algorithm we will distinguish between \emph{heavy} and \emph{light} vertices, depending on how large the out- and in-balls of these vertices are. To classify vertices as heavy or light, we rely on the following simple lemmas:

\begin{lemma}[Estimate Ball Sizes] \label{lem:approx-ball-size}
Let $\varepsilon > 0$. Given a directed graph $G$ with nonnegative edge weights and $r > 0$, we can approximate $|\Bout(v, r)|$ with additive error $\varepsilon n$ for each vertex~$v$. With high probability, the algorithm succeeds and runs in time $\Order(\varepsilon^{-2} \log n  \cdot (m + n \log\log n))$.
\end{lemma}
\begin{proof}
Sample random vertices $v_1,\ldots,v_k \in V(G)$ (with repetition) for $k := 5 \varepsilon^{-2} \log n$. Compute $\Bin(v_i, r)$ for all~$i \in [k]$. Using Dijkstra's algorithm with Thorup's priority queue~\cite{Dijkstra59,Thorup03}, this step runs in time $\Order(k \cdot (m + n \log\log n)) = \Order(\varepsilon^{-2} \log n \cdot (m + n \log\log n))$. Now return for each vertex $v$, the estimate
\begin{equation*}
    b(v) := \frac{n}{k} \cdot |\set{i \in [k]: v \in \Bin(v_i, r)}|.
\end{equation*}

We claim that this estimate is accurate. Let $I_i$ denote the indicator variable whether $v_i \in \Bout(v,r)$, and let $I := \sum_{i=1}^k I_i$. Then the random variable $b(v)$ is exactly
\begin{equation*}
    b(v) = \frac{n}{k} \cdot I.
\end{equation*}
Note that $\Pr(I_i = 1) = |\Bout(v,r)|/n$.
In expectation we therefore have 
\begin{equation*}
    \Ex(b(v)) = \frac{n}{k} \cdot \sum_{i=1}^k \Pr(I_i = 1) = \frac nk \cdot \frac{k}{n} \cdot |\Bout(v, r)| = |\Bout(v, r)|. 
\end{equation*}
Using Chernoff's bound we have $\Pr(|I - \Ex(I)| > a) < 2 \exp(-2a^2/k)$. For $a := \varepsilon k$ we obtain 
\begin{equation*}
    \Pr(|b(v) - \Ex(b(v))| > \varepsilon n) = \Pr(|I - \Ex(I)| > \varepsilon k) < 2 \exp(-2\epsilon^2 k) \le 2 n^{-10}.
\end{equation*}
Hence, with high probability the computed estimates are accurate.
\end{proof}

\begin{lemma}[Heavy/Light Classification] \label{lem:ldd-heavy-light}
There is an algorithm $\textsc{Light}(G, r)$ that, given a directed graph $G$ and a radius $r$, returns a set $L \subseteq V(G)$ with the following properties:
\begin{itemize}
    \smallskip
    \item For all $v \in L$, it holds that $|\Bout(v, r)| \leq \frac 78 n$.
    \item For all $v \in V(G) \setminus L$, it holds that $|\Bout(v, r)| \geq \frac 34 n$.
    \item $\textsc{Light}(G, r)$ runs in time $\Order((m + n \log \log n) \log n)$.
\end{itemize}
\end{lemma}
\begin{proof}
We run the previous \cref{lem:approx-ball-size} with parameter $\epsilon := \frac 1{16}$, and let $L$ be the subset of vertices with estimated ball sizes at most $\frac{13}{16} n$. With high probability, the estimates have additive error at most $\epsilon n = \frac1{16} n$. Therefore any vertex $v \in L$ satisfies $|\Bout(v, r)| \leq \frac{13}{16} n + \frac{1}{16}n = \frac 78 n$ and any vertex $v \in V(G) \setminus L$ satisfies $|\Bout(v, r)| \geq \frac{13}{16} n - \frac{1}{16} n = \frac 34 n$. The running time is dominated by \cref{lem:approx-ball-size} which runs in time $\Order((m + n \log \log n) \log n)$ as claimed.
\end{proof}

\subsection{The Strong Low-Diameter Decomposition}
The strong LDD works as follows: Let~\smash{$R = \frac{D}{10\log n}$}. First, we run \cref{lem:ldd-heavy-light} on $G$ with radius $R$ to compute a set $L^\OUT$ and we run \cref{lem:ldd-heavy-light} on the reversed graph with radius $R$ to compute a set $L^\IN$. We refer to the vertices in~$L^\OUT$ as \emph{out-light}, to the vertices in $L^\IN$ as \emph{in-light}, and to the vertices in $V(G) \setminus (L^\OUT \cup L^\IN)$ as \emph{heavy}. Then we distinguish two cases:
\medskip
\begin{itemize}
\item\emph{The heavy case:} If there is a heavy vertex $v \in V(G) \setminus (L^\OUT \cup L^\IN)$, we compute the set of vertices $W$ that both reach $v$ and are reachable from $v$ within distance $R$, i.e., $W = \Bout(v, R) \cap \Bin(v, R)$. Let $T^\OUT, T^\IN$ denote the shortest path trees from $v$ to $W$ and from $W$ to $v$, respectively. Let $C$ be the union of vertices in $T^\OUT$ and $T^\IN$. We \emph{collapse} $C$ (that is, we replace all vertices in $C$ by a single super-vertex) and consider the remaining (multi-)graph $G / C$. We recursively compute the strong LDD in $G / C$, resulting in a set of edges $S$. In $S$ we uncollapse all edges involving the super-vertex (i.e., for any edge $(v, u) \in E(G)$ which became an edge $(C, u)$ in the collapsed graph, we revert~$(C, u)$ back to $(v, u)$) and return $S$.
\medskip
\item\emph{The light case:} If there is no heavy vertex, then each vertex is out-light or in-light. For each vertex $v$ (which is out-light, say) we can therefore proceed in the standard way: Sample a radius $r$ from a geometric distribution with parameter $\Order(\log n / D)$, cut the edges leaving $\Bout(v, r)$ and recur on both the inside and the outside of the ball $\Bout(v, r)$.
\end{itemize}
\medskip
We summarize the pseudocode with the precise parameters in \cref{alg:strong-ldd}. Throughout this section, we denote by $n_0$ the size of the original graph and by $n$ the size of the current graph~$G$ (in the current recursive call of the algorithm).

\begin{algorithm}[t]
\caption{The implementation of the strong Low-Diameter Decomposition (see \cref{thm:strong-ldd}) that either returns a set of edges $S \subseteq E(G)$ or $\textsc{Fail}$.} \label{alg:strong-ldd}
\smallskip
\begin{algorithmic}[1]
\Procedure{StrongLDD}{$G, D$}
    \If{$n \leq 100$} \Return $E(G)$ \EndIf
    \State Let $R := \frac{D}{10 \log n}$
    \State Compute $L^\OUT \gets \textsc{Light}(G, R)$
    \State Compute $L^\IN \gets \textsc{Light}(G^{\mathit{rev}}, R)$~~\emph{(here, $G^{\mathit{rev}}$ is the graph $G$ with reversed edges)}

    \medskip
    \LineComment{The heavy case}
    \If{there is a heavy vertex $v \in V(G) \setminus (L^\OUT \cup L^\IN)$}
        \State Let $W \gets \Bout(v, R) \cap \Bin(v, R)$
        \State Compute shortest path trees $T^\OUT$ from $v$ to $W$, and $T^\IN$ from $W$ to $v$
        \State Let $C$ be the union of vertices in $T^\OUT, T^\IN$
        \State $S \gets \textsc{StrongLDD}(G / C, D - 4R)$
        \State \Return $S$ after uncollapsing all edges 
    \EndIf

    \medskip
    \LineComment{The light case}
    \State $S \gets \emptyset$
    \While{there is $v \in V(G) \cap L^\OUT$} \label{alg:strong-ldd:line:loop-out}
        \State $r \sim \Geom(R^{-1} \cdot 10 \log n_0)$
        \If{$r > R$} \Return $\textsc{Fail}$\EndIf
        \State $S \gets S \cup \partial \Bout(v, r) \cup \textsc{StrongLDD}(G[\Bout(v, r)], D)$
        \State $G \gets G \setminus \Bout(v, r)$
    \EndWhile
    \While{there is $v \in V(G) \cap L^\IN$} \label{alg:strong-ldd:line:loop-in}
        \State $r \sim \Geom(R^{-1} \cdot 10 \log n_0)$
        \If{$r > R$} \Return $\textsc{Fail}$ \EndIf
        \State $S \gets S \cup \partial \Bin(v, r) \cup \textsc{StrongLDD}(G[\Bin(v, r)], D)$
        \State $G \gets G \setminus \Bin(v, r)$
    \EndWhile
    \State \Return $S$
\EndProcedure
\end{algorithmic}
\end{algorithm}

\begin{lemma}[Strong Diameter of \cref{alg:strong-ldd}] \label{lem:strong-ldd-correctness}
With high probability, $\textsc{StrongLDD}(G, D)$ either returns $\textsc{Fail}$ or a set of edges $S \subseteq E(G)$ such that every strongly connected component~$C$ of $G \setminus S$ has diameter at most $D$, i.e., $\max_{u, v \in C} \dist_{G[C]}(u, v) \leq D$.
\end{lemma}
\begin{proof}
With high probability, the heavy-light classification works correctly in the execution of $\textsc{StrongLDD}(G, D)$ (and all recursive calls). We condition on this event and treat the classification as perfect.

As before, we have to distinguish the heavy and the light case. In the heavy case, let~$v$ be the heavy vertex and let $W, T^\OUT, T^\IN, C$ be as in the algorithm. We claim that the induced subgraph $G[C]$ has diameter at most $4R$. Take any vertex $x \in C$; it suffices to prove that both $\dist_{G[C]}(v, x) \leq 2R$ and $\dist_{G[C]}(x, v) \leq 2R$. We show the former claim and omit the latter. There are two easy cases: Either we have $x \in T^\OUT$ in which case we immediately have that $\dist_{G[C]}(v, x) \leq R$ (as any path in $T^\OUT$ has length at most $R$). Or we have $x \in T^\IN$, in which case there exists some intermediate vertex $y \in W$ with $\dist_{G[C]}(y, x) \leq R$. But then also $\dist_{G[C]}(v, y) \leq R$ and in combination we obtain $\dist_{G[C]}(v, x) \leq 2R$ as claimed.

Recall that the algorithm collapses the vertices in $C$, and computes a strong LDD~$S$ on the remaining multigraph with parameter $D - 4R$. We assume by induction that the recursive call computes a correct strong decomposition (for $G / C$). To see that the decomposition is also correct for $G$, take any two vertices $u, v$ in the same strongly connected component in $G \setminus S$. We have that $\dist_{(G / C) \setminus S}(u, v) \leq D - 4R$. If the shortest $u$-$v$-path in $G / C$ does not touch the supervertex, then we immediately have $\dist_{G \setminus S}(u, v) \leq D - 4R \leq D$. If the shortest path touches the supervertex, then we can replace the path through $C$ by a path of length $\diam(G[C]) \leq 4R$. It follows that $\dist_{G \setminus S}(u, v) \leq D - 4R + 4R \leq D$.

The correctness of the light case is exactly as in the known LDD by~\cite{BernsteinNW22}, and similar to \cref{lem:decomposition}: For every ball $\Bout(v, r)$ (or $\Bin(v, r)$) that the algorithm carves out, we remove all outgoing edges $\partial \Bout(v, r)$ (or all incoming edges $\partial \Bin(v, r)$, respectively). Thus, two vertices $x, y$ in the remaining graph are part of the same strongly connected component only if both $x, y \in \Bout(v, r)$ or both~\makebox{$x, y \not\in \Bout(v, r)$}. The algorithm continues the loop on all vertices outside $\Bout(v, r)$ and recurs inside $\Bout(v, r)$. By induction, both calls succeed and reduce the diameter to at most $D$.

Eventually the algorithm reaches a base case where $G$ contains only a constant number of nodes and edges---in this case, we can select $S$ to be the whole set of edges.
\end{proof}

\begin{lemma}[Sparse Hitting of \cref{alg:strong-ldd}] \label{lem:strong-ldd-sparse-hitting}
For any edge $e \in E(G)$, the probability that $e$ is contained in the output of $\textsc{StrongLDD}(G, D)$ is at most $\Order(\frac{w(e)}D \cdot \log^3(n_0) + \frac{1}{\poly(n)})$.
\end{lemma}
\begin{proof}
In this proof we condition on the event that the initially computed heavy/light classification is correct. Since this event happens with high probability, we only increase the hitting probabilities by $\frac{1}{\poly(n)}$ for all edges.

Let $p(n, w, D)$ be an upper bound on the probability that an edge of weight $w$ is contained in the output of $\textsc{StrongLDD}(G, D)$, where $G$ is an $n$-vertex graph. We inductively prove that $p(n, w, D) \leq \frac{w}D \cdot 1000 \log(n_0) \log^2(n)$ which is as claimed. We distinguish the heavy and light case in \cref{alg:strong-ldd}.

\proofsubparagraph{The Light Case.}
Suppose that the algorithm enters the light case (that is, there is no vertex classified as heavy). Focus on some edge $e = (x, y)$ of weight $w = w(e)$. We distinguish three cases for each iteration. Suppose that the current iteration selects an out-light vertex $v$.
\smallskip
\begin{itemize}
    \item $x, y \not\in \Bout(v, r)$: The edge $e$ is not touched in this iteration and remains a part of the graph $G$. It may or may not be included in the output, depending on the future iterations.
    \item $x \in \Bout(v, r)$ and $y \not\in \Bout(v, r)$: In this case $e \in \partial\Bout(v, r)$ and thus the edge is included into $S$.
    \item $y \in \Bout(v, r)$: The edge is not included in $\partial\Bout(v, r)$. It may however be included in the recursive call on $\Bout(v, r)$. In the recursive call we have that $|\Bout(v, r)| \leq |\Bout(v, R)| \leq \frac{7n}8$, as $r \leq R$ (in the opposite case the algorithm fails and no edge is returned) and by \cref{lem:ldd-heavy-light} as $v$ is out-light.
\end{itemize}
\smallskip
Combining these cases, we obtain the following recursion for $p(n, w, D)$. In the calculation we abbreviate $q := R^{-1} \cdot 10\log(n_0)$:
\begin{gather*}
    p(n, w, D) \leq \max_{v \in V(G)} \Pr_{r \sim \Geom(q)} (y \not\in \Bout(v, r) \mid x \in \Bout(v, r)) + p(\tfrac{7n}8, w, D)\\
    \qquad\leq \max_{v \in V(G)} \Pr_{r \sim \Geom(q)} (r < \dist(v, y) \mid r \geq \dist(v, x)) + p(\tfrac{7n}8, w, D) \\
    \qquad\leq \max_{v \in V(G)} \Pr_{r \sim \Geom(q)} (r < \dist(v, x) + w \mid r \geq \dist(v, x)) + p(\tfrac{7n}8, w, D)
\intertext{Let $r' := r - \dist(v, x)$. Conditioned on the event $r \geq \dist(v, x)$, $r'$ is a nonnegative random random variable and by the memoryless property of geometric distributions, $r'$ is sampled from $\Geom(q)$, too:}
    \qquad\leq \max_{v \in V(G)} \Pr_{r' \sim \Geom(q)}(r < w) + p(\tfrac{7n}8, w, D) \\
    \qquad\leq wq + p(\tfrac{7n}8, w, D) \\
    \qquad\leq \frac{w}{D} \cdot 100\log(n_0) \log(n) + p(\tfrac{7n}8, w, D).
\end{gather*}
In the last step, we have plugged in $q = R^{-1} \cdot 10 \log(n_0) = \frac1D \cdot 100 \log(n_0) \log(n)$. It follows by induction that $p(n, w, D) \leq \frac wD \cdot 100 \log(n_0) \log(n) \log_{8/7}(n) \leq \frac wD \cdot 1000 \log(n_0) \log^2(n)$.

The same analysis applies also to the in-balls with ``$\Bin$'' in place of ``$\Bout$''.

\proofsubparagraph{The Heavy Case.} In the heavy case, the algorithm selects a heavy vertex $v$, computes the sets~$W = \Bout(v, R) \cap \Bin(v, R)$ and $C \supseteq W$ and recurs on the graph $G / C$ in which we contract the vertex set $C$ to a single vertex. We have $|\Bout(v, R)|, |\Bin(v, R)| > \frac{3n}4$ by \cref{lem:ldd-heavy-light} since $v$ is heavy. It follows that $|C| \geq |W| > \frac n2$ and therefore the contracted graph has size $|V(G / C)| \leq \frac n2$. As we call the algorithm recursively with parameter $D - 4R$ where~$R = \frac{D}{10 \log n}$, we obtain the following recurrence:
\begin{gather*}
    p(n, w, D) \leq p(\tfrac n2, w, D - 4R).
\end{gather*}
Using the induction hypothesis, we obtain:
\begin{gather*}
    p(n, w, D) \\
    \qquad\leq \frac{w}{D - 4R} \cdot 1000 \log^2(n_0) \log(\tfrac n2) \\
    \qquad\leq \frac{w}{D} \cdot \frac{1}{1 - \frac{4}{10 \log n}}\cdot 1000 \log^2(n_0) \log(\tfrac n2) \\
    \qquad= \frac{w}{D} \cdot \frac{\log(n)}{\log(n) - \frac{4}{10}} \cdot 1000 \log^2(n_0) \cdot (\log(n) - 1) \\
    \qquad\leq \frac{w}{D} \cdot 1000 \log^2(n_0) \cdot \log(n). \qedhere
\end{gather*}
\end{proof}

\begin{lemma}[Running Time of \cref{alg:strong-ldd}] \label{lem:strong-ldd-time}
The algorithm $\textsc{StrongLDD}(G, D)$ runs in time $\Order((m + n_0 \log \log n_0) \log^2(n_0))$.
\end{lemma}
\begin{proof}
First focus on a single call of the algorithm and ignore the cost of recursive calls. It takes time $\Order((m + n_0 \log \log n_0) \log(n_0))$ to compute the heavy-light classification. In the heavy case, we can compute $W, T^\OUT, T^\IN, C$ in Dijkstra-time $\Order(m + n_0 \log\log n_0)$. In the light case, we can also carve out all balls $\Bout(v, r)$ and $\Bin(v, r)$ in total time $\Order(m  + n_0 \log\log(n_0))$, although the formal analysis is more involved: Observe that we explore each vertex at most once spending time $\Order(\log\log n_0)$, and that we explore each edge at most once spending time~$\Order(1)$. Since the analysis is similar to \cref{lem:decomposition}, we omit further details.

As the algorithm recurs on disjoint subgraphs of $G$, where the number of nodes in each subgraph is a constant factor smaller than the original number of nodes or less, the running time becomes $\Order((m + n_0 \log \log n_0) \log(n_0)^2)$.
\end{proof}

\begin{lemma}[Failure Probability of \cref{alg:strong-ldd}] \label{lem:strong-ldd-fail}
$\textsc{StrongLDD}(G, D)$ returns $\textsc{Fail}$ with probability at most $\Order(n_0^{-8})$.
\end{lemma}
\begin{proof}
As shown in detail in the previous lemmas, with every recursive call the number of vertices reduces by a constant factor and thus the recursion reaches depth at most~$\Order(\log n_0)$. In each recursive call, the loops in \cref{alg:strong-ldd:line:loop-out,alg:strong-ldd:line:loop-in} run at most $n_0$ times. For each execution, the error event is that $r > R$, where $r \sim \Geom(R^{-1} \cdot 10 \log(n_0))$. This event happens with probability at most $\exp(-10 \log(n_0)) \leq n_0^{-10}$, and therefore the algorithm returns $\textsc{Fail}$ with probability at most $\Order(n_0 \log n_0)\cdot n_0^{-10} \leq \Order(n_0^{-8})$.
\end{proof}

\begin{proof}[Proof of \cref{thm:strong-ldd}]
To compute the claimed strong LDD we call $\textsc{StrongLDD}(G, \frac12 D)$ with the following two modifications:

First, whenever some recursive call returns $\textsc{Fail}$, we simply restart the whole algorithm.

Second, we test whether the returned set of edges $S \subseteq E(G)$ satisfies the Strong Diameter property. To this end, we compute the strongly connected components in $G \setminus S$ and compute, for any such component $C$, a $2$-approximation of its diameter. By  a standard argument, such a $2$-approximation can be obtained in Dijkstra-time by (1) selecting an arbitrary node $v$, (2) computing $d^{out} := \max_{u \in V(G)} d_G(v,u)$ by solving SSSP on $G$, (3) computing $d^{in} := \max_{u \in V(G)} d_G(u,v)$ by solving SSSP on the reversed graph of $G$, and returning $\max\{d^{in}, d^{out}\}$. If the diameter approximations are at most $\frac D2$ in all components, we return~$S$. Otherwise, we restart the whole algorithm.

This algorithm indeed never fails to satisfy the Strong Diameter property: Since the diameter approximations have approximation factor at most $2$, we have certified that the diameter of any strongly connected component is at most $D$ in the graph $G \setminus S$. Moreover, with high probability the execution of \cref{alg:strong-ldd} passes both tests with high probability (by \cref{lem:strong-ldd-fail,lem:strong-ldd-correctness}), and therefore we expect to repeat the algorithm $\Order(1)$ times. Since the repetitions are independent of each other, the edge hitting probability increases only by a constant factor and remains $\Order(\frac{w(e)}{D} \cdot \log^3(n_0))$ by \cref{lem:strong-ldd-sparse-hitting}.

Finally, consider the running time. As argued before, with high probability we avoid restarting \cref{alg:strong-ldd} altogether. Thus, with high probability the algorithm runs in total time is $\Order((m + n_0 \log \log n_0) \log^2 (n_0))$ by \cref{lem:strong-ldd-time}. Since we expect to repeat the algorithm at most $\Order(1)$ times, the same bound applies to the expected running time.
\end{proof}

\appendix
\section{Lazy Dijkstra} \label{sec:dijkstra}
This section is devoted to a proof of the following lemma, stating that Dijkstra's algorithm can be adapted to work with negative edges in time depending on the $\eta_G(v)$ values. Recall that $\eta_G(v)$ denotes the minimum number of negative-weight edges in a shortest $s$-$v$ path in~$G$.

\lemlazydijkstra*

This lemma is basically \cite[Lemma~3.3]{BernsteinNW22}, but the statement differs slightly. We provide a self-contained proof that morally follows the one in~\cite[Appendix~A]{BernsteinNW22}.

We give the pseudocode for \cref{lem:lazy-dijkstra} in \cref{alg:lazy-dijkstra}. Throughout, let $G = (V, E, w)$ be the given directed weighted graph with possibly negative edge weights. We write $E^{\geq 0}$ for the subset of edges with nonnegative weight, and $E^{< 0}$ for the subset of edges with negative weight. In the pseudocode, we rely on Thorup's priority queue:

\begin{lemma}[Thorup's Priority Queue~\cite{Thorup03}] \label{lem:thorup}
There is a priority queue implementation for storing $n$ integer keys that supports the operations $\textsc{FindMin}$, $\textsc{Insert}$ and $\textsc{DecreaseKey}$ in constant time, and $\textsc{Delete}$ in time $\Order(\log\log n)$.
\end{lemma}

\begin{algorithm}[t]
\caption{The version of Dijkstra's algorithm implementing \cref{lem:lazy-dijkstra}.} \label{alg:lazy-dijkstra}
\begin{algorithmic}[1]
\State Initialize $d[s] \gets 0$ and $d[v] \gets \infty$ for all vertices $v \neq s$
\State Initialize a Thorup priority queue $Q$ with keys $d[\cdot]$ and add $s$ to $Q$
\Repeat
\medskip
\Statex[1]\emph{(The Dijkstra phase)}
\State $A \gets \emptyset$
\While{$Q$ is nonempty}
    \State Remove the vertex $v$ from $Q$ with minimum $d[v]$
    \State Add $v$ to $A$
    \ForEach{edge $(v, x) \in E^{\geq 0}$}
        \If{$d[v] + w(v, w) < d[x]$} \label{alg:lazy-dijkstra:line:dijkstra-relax-start}
            \State Add $x$ to $Q$
            \State $d[x] \gets d[v] + w(v, x)$ \label{alg:lazy-dijkstra:line:dijkstra-relax-end}
        \EndIf
    \EndForEach
\EndWhile

\medskip
\Statex[1]\emph{(The Bellman-Ford phase)}
\ForEach{$v \in A$}
    \ForEach{edge $(v, x) \in E^{< 0}$}
        \If{$d[v] + w(v, x) < d[x]$} \label{alg:lazy-dijkstra:line:bf-relax-start}
            \State Add $x$ to $Q$
            \State $d[x] \gets d[v] + w(v, x)$ \label{alg:lazy-dijkstra:line:bf-relax-end}
        \EndIf
    \EndForEach
\EndForEach

\medskip
\Until{$Q$ is empty}
\State\Return $d[v]$ for all vertices $v$
\end{algorithmic}
\end{algorithm}

For the analysis of the algorithm, we define two central quantities. Let $v$ be a vertex, then we define
\begin{align*}
    \dist_i(v) &= \min\set{ w(P) : \text{$P$ is an $s$-$v$-path containing less than $i$ negative edges}}, \\
    \dist'_i(v) &= \min\left\{ \dist_i(v), \min_{\substack{u \in V\\w(u, v) < 0}} \dist_i(u) + w(u, v) \right\}.
\end{align*}
Note that $\dist_0(v) = \dist'_0(v) = \infty$. We start with some observations involving these quantities $\dist_i$ and $\dist_i'$:

\begin{observation} \label{obs:1}
For all $i$, $\dist_i(v) \geq \dist'_i(v) \geq \dist_{i+1}(v)$.
\end{observation}

\begin{observation} \label{obs:2}
For all $v$,
\begin{equation*}
    \dist_{i+1}(v) = \min\left\{\dist_i(v), \min_{\substack{u \in V\\\dist_i(u) > \dist'_i(u)}} \dist'_i(u) + \dist_{G^{\geq 0}}(u, v)\right\}.
\end{equation*}
\end{observation}
\begin{proof}
The statement is clear if $\dist_{i}(v) = \dist_{i+1}(v)$, so assume that $\dist_{i+1}(v) < \dist_i(v)$. Let $P$ be the path witnessing $\dist_{i+1}(v)$, i.e., a shortest $s$-$v$-path containing less than~$i+1$ negative edges. Let $(x, u)$ denote the last negative-weight edge in $P$, and partition the path~$P$ into subpaths $P_1\, x\, u\, P_2$. Then the first segment $P_1\, x$ is a path containing less than~$i$ negative-weight edges and the segment $u\, P_2$ does not contain any negative-weight edges. Therefore,
\begin{equation*}
    \dist_{i+1}(v) = \dist_i(x) + w(x, u) + \dist_{G^{\geq 0}}(u, v) \geq \dist_i'(u) + \dist_{G^{\geq 0}}(u, v).
\end{equation*}
Suppose, for the sake of contradiction, that $\dist_i(u) = \dist'_i(u)$. Then
\begin{equation*}
    \dist_{i+1}(v) \geq \dist_i(u) + \dist_{G^{\geq 0}}(u, v) \geq \dist_i(v),
\end{equation*}
which contradicts our initial assumption.
\end{proof}

\begin{observation} \label{obs:3}
For all $v$,
\begin{equation*}
    \dist'_i(v) = \min\left\{\dist_i(v), \min_{\substack{u \in V\\\dist_{i-1}(u) > \dist_i(u)\\w(u, v) < 0}} \dist_i(u) + w(u, v) \right\}
\end{equation*}
\end{observation}
\begin{proof}
The statement is clear if $\dist_i(v) = \dist_i'(v)$, so suppose that $\dist_i'(v) < \dist_i(v)$. Then there is some vertex $u \in V$ with $w(u, v) < 0$ such that $\dist_i'(v) = \dist_i(u) + w(u, v)$. It suffices to prove that $\dist_{i-1}(u) > \dist_i(u)$. Suppose for the sake of contradiction that $\dist_{i-1}(u) = \dist_i(u)$. Then $\dist_i'(v) = \dist_{i-1}(u) + w(u, v) \geq \dist'_{i-1}(v)$, which contradicts our initial assumption (by \cref{obs:1}).
\end{proof}

\begin{lemma}[Invariants of \cref{alg:lazy-dijkstra}] \label{lem:lazy-dijkstra-invariants}
Consider the $i$-th iteration of the loop in \cref{alg:lazy-dijkstra} (starting at $1$). Then the following invariants hold:
\smallskip
\begin{enumerate}[1.]
    \item After the Dijkstra phase (after \cref{alg:lazy-dijkstra:line:dijkstra-relax-end}):
    \begin{enumerate}[a.]
        \item $d[v] = \dist_i(v)$ for all vertices $v$, and
        \item $A = \set{v : \dist_{i-1}(v) > \dist_i(v)}$.
    \end{enumerate}
    \smallskip
    \item After the Bellman-Ford phase (after \cref{alg:lazy-dijkstra:line:bf-relax-end}):
    \begin{enumerate}
        \item $d[v] = \dist'_i(v)$ for all vertices $v$, and
        \item $Q = \set{ v : \dist_i(v) > \dist'_i(v) }$.
    \end{enumerate}
\end{enumerate}
\end{lemma}
\begin{proof}
We prove the invariants by induction on $i$.

\proofsubparagraph{First Dijkstra Phase}
We start with the analysis of the first iteration, $i = 1$. The execution of the Dijkstra phase behaves exactly like the regular Dijkstra algorithm. It follows that $d[v] = \dist_{G^{\geq 0}}(s, v) = \dist_1(v)$, as claimed in Invariant 1a. Moreover, we include in $A$ exactly all vertices which were reachable from $s$ in $G^{\geq 0}$. Indeed, for these vertices $v$ we have that $\dist_1(v) = \dist_{G^{\geq 0}}(s, v) < \infty$ and $\dist_0(v) = \infty$, and thus $A = \set{ v : \dist_0(v) > \dist_1(v) }$, which proves Invariant 1b.

\proofsubparagraph{Later Dijkstra Phase}
Next, we analyze the Dijkstra phase for a later iteration, $i > 1$. Letting $d'$ denote the state of the array $d$ after the Dijkstra phase, our goal is to prove that $d'[v] = \dist_i(v)$ for all vertices $v$. So fix any vertex $v$; we may assume that $\dist_{i+1}(v) < \dist_i(v)$, as otherwise the statement is easy using that the algorithm never increases $d[\cdot]$. A standard analysis of Dijkstra's algorithm reveals that
\begin{equation*}
    d'[v] = \min_{u \in Q} (d[u] + \dist_{G^{\geq 0}}(u, v)),
\end{equation*}
where $Q$ is the queue before the execution of Dijkstra. By plugging in the induction hypothesis and \cref{obs:2}, we obtain that indeed
\begin{equation*}
    d'[v] = \min_{\substack{u \in V\\\dist_{i-1}(v) > \dist'_{i-1}(v)}} d[u] + \dist_{G^{\geq 0}}(u, v) = \dist_i(v),
\end{equation*}
which proves Invariant 1a.

To analyze Invariant 1b and the set $A$, first recall that we reset $A$ to an empty set before executing the Dijkstra phase. Afterwards, we add to $A$ exactly those vertices that are either (i) contained in the queue $Q$ initially or (ii) for which $d'[v] < d[v]$. Note that these sets are exactly (i) $\set{ v : \dist_i(v) > \dist'_i(v) }$ and (ii) $\set{v : \dist'_{i-1}(v) > \dist_i(v)}$ whose union is exactly $\set{v : \dist_{i-1}(v) > \dist_i(v)}$ by \cref{obs:1}.

\proofsubparagraph{Bellman-Ford Phase}
The analysis of the Bellman-Ford phase is simpler. Writing again~$d'$ for the state of the array $d$ after the execution of the Bellman-Ford phase, by \cref{obs:3} we have that
\begin{equation*}
    d'[v] = \min_{\substack{u \in A\\w(u, v) < 0}} d[u] + w(u, v) = \min_{\substack{u \in V\\\dist'_{i-1}(u) > \dist_i(u)\\w(u, v) < 0}} \dist_i(u) + w(u, v) = \dist_i'(v),
\end{equation*}
which proves Invariant 2a. Here again we have assumed that $\dist_i'(v) < \dist_i(v)$, as otherwise the statement is trivial since the algorithm never increases $d[\cdot]$.

Moreover, after the Dijkstra phase has terminated, the queue $Q$ was empty. Afterwards, in the current Bellman-Ford phase, we have inserted exactly those vertices $v$ into the queue for which $\dist_i(v) > \dist_i'(v)$ and thus $Q = \set{ v : \dist_i(v) > \dist_i'(v)}$, which proves Invariant~2b.
\end{proof}

From these invariants (and the preceding observations), we can easily conclude the correctness of \cref{alg:lazy-dijkstra}:

\begin{lemma}[Correctness of \cref{alg:lazy-dijkstra}] \label{lem:lazy-dijkstra-correctness}
If the given graph $G$ contains a negative cycle, then \cref{alg:lazy-dijkstra} does not terminate. Moreover, if \cref{alg:lazy-dijkstra} terminates, then it has correctly computed $d[v] = \dist_G(s, v)$.
\end{lemma}
\begin{proof}
We show that after the algorithm has terminated, all edges $(u, v)$ are \emph{relaxed}, meaning that~$d[v] \leq d[u] + w(u, v)$. Indeed, suppose there is an edge $(u, v)$ which is not relaxed, i.e., $d[v] > d[u] + w(u, v)$. Let $i$ denote the final iteration of the algorithm. By Invariant~2a we have that $d[x] = \dist_i'(x)$ and by Invariant~2b we have that $\dist_i'(x) = \dist_i(x)$ (assuming that $Q = \emptyset$), for all vertices $x$. We distinguish two cases: If $w(u, v) \geq 0$, then we have that $\dist_i(v) > \dist_i(u) + w(u, v)$---a contradiction. And if $w(u, v) < 0$, then we have that $\dist_i'(v) = \dist_i(u) + w(u, v)$---also a contradiction.

So far we have proved that if the algorithm terminates, all edges are relaxed. It is easy to check that if $G$ contains a negative cycle, then at least one edge in that cycle cannot be relaxed. It follows that the algorithm does not terminate whenever $G$ contains a negative cycle.

Instead, assume that $G$ does not contain a negative cycle. We claim that the algorithm has correctly computed all distances. First, recall that throughout we have $d[v] \geq \dist_G(s, v)$. Consider any shortest $s$-$v$-path~$P$; we prove that $d[v] = w(P)$ by induction on the length of $P$. For $|P| = 0$, we have correctly set $d[s] = 0$ initially. (Note that $\dist_G(s, s)$ cannot be negative as otherwise $G$ would contain a negative cycle.) So assume that $P$ is nonempty and that $P$ can be written as $P_1\, u\, v$. Then by induction $d[u] = \dist_G(P_1\, u)$. Since the edge $(u, v)$ is relaxed, we have that $d[v] \leq d[u] + w(u, v) = w(P) = \dist_G(s, v)$. Recall that we also have $d[v] \geq \dist_G(s, v)$ and therefore $d[v] = \dist_G(s, v)$.
\end{proof}

For us, the most relevant change in the proof is the running time analysis. Recall that~$\eta_G(v)$ denotes the minimum number of negative edges in a shortest $s$-$v$-path, and that $\deg(v)$ denotes the out-degree of a vertex $v$.

\begin{lemma}[Running Time of \cref{alg:lazy-dijkstra}] \label{lem:lazy-dijkstra-time}
Assume that $G$ does not contain a negative cycle. Then \cref{alg:lazy-dijkstra} runs in time $\Order(\sum_v (\deg(v) + \log\log n) \eta_G(v))$.
\end{lemma}
\begin{proof}
Consider a single iteration of the algorithm. Letting $A$ denote the state of the set $A$ at the end of (Dijkstra's phase of) the iteration, the running time of the whole iteration can be bounded by:
\begin{equation*}
    \Order\left(\sum_{v \in A} (\deg(v) + \log\log n)\right).
\end{equation*}
Indeed, in the Dijkstra phase, in each iteration we spend time $\Order(\log\log n)$ for deleting an element from the queue (\cref{lem:thorup}), but for each such deletion in $Q$ we add a new element to $A$. Moreover, both in the Dijkstra phase and the Bellman-Ford phase we only enumerate edges starting from a vertex in $A$, amounting for a total number of $\Order(\sum_{v \in A} \deg(v))$ edges. The inner steps of the loops (in \crefrange{alg:lazy-dijkstra:line:dijkstra-relax-start}{alg:lazy-dijkstra:line:dijkstra-relax-end} and \crefrange{alg:lazy-dijkstra:line:bf-relax-start}{alg:lazy-dijkstra:line:bf-relax-end}) run in constant time each (\cref{lem:thorup}).

Let us write $A_i$ for the state of $i$ in the $i$-th iteration. Then the total running time is
\begin{equation*}
    \Order\left(\sum_{i=1}^\infty \sum_{v \in A_i} (\deg(v) + \log\log n)\right) = \Order\left(\sum_{v \in V} |\set{i : v \in A_i}| \cdot (\deg(v) + \log\log n)\right).
\end{equation*}
To complete the proof, it suffices to show that $|\set{i : v \in A_i}| \leq \eta_G(v)$. To see this, we first observe that $\dist_{\eta_G(v) + 1}(v) = \dist_{\eta_G(v) + 2} = \dots = \dist_G(s, v)$. Since, by the invariants above we know that $A_i = \set{ v : \dist_{i-1}(v) > \dist_i(v)}$, it follows that $v$ can only be contained in the sets $A_1, \dots, A_{\eta_G(v)}$.
\end{proof}

In combination, \cref{lem:lazy-dijkstra-correctness,lem:lazy-dijkstra-time} complete the proof of \cref{lem:lazy-dijkstra}.

\bibliographystyle{plainurl}
\bibliography{refs}

\end{document}